\documentclass[a4paper,onecolumn,10pt]{quantumarticle}
\pdfoutput=1

\usepackage[numbers]{natbib}
\usepackage{soul,color}
\usepackage{xcolor}
\usepackage{graphicx}
\usepackage[shortlabels]{enumitem}

\usepackage{algorithm}
\usepackage{algpseudocode}
\algrenewcommand\algorithmicrequire{\textbf{Input:}}
\algrenewcommand\algorithmicensure{\textbf{Output:}}
\newfloat{subroutine}{tbp}{los}
\floatname{subroutine}{Subroutine}

\usepackage{tikz}
\usetikzlibrary{quantikz2}
\usepackage{mathtools}
\usepackage{amsfonts}
\usepackage{amsthm}
\usepackage{amssymb}
\usepackage{amsmath}
\usepackage{dcolumn}
\usepackage{physics}
\usepackage{float}
\usepackage{bm}
\usepackage{verbatim}
\usepackage{braket}
\usepackage{bbm}
\usepackage{cancel}
\newtheorem{theorem}{Theorem}
\newtheorem*{theorem*}{Theorem}

\theoremstyle{remark}

\DeclareMathOperator*{\argmax}{arg\,max}

\newcommand{\supp}{\operatorname{supp}}
\newcommand{\poly}{\operatorname{poly}}
\usepackage[colorlinks,linkcolor=blue,anchorcolor=blue,citecolor=blue,urlcolor=blue]{hyperref}

\begin{document}
\title{Polynomial-time exact diagonalization via sparse guided eigenwalks}
\author{Zachary E. Chin}\email{zchin@mit.edu}\affiliation{
IBM Quantum, IBM T.J. Watson Research Center, Yorktown Heights, NY 10598, USA}\affiliation{Department of Physics, Massachusetts Institute of Technology, Cambridge, Massachusetts 02139, USA}
\author{Mario Motta}\affiliation{
IBM Quantum, IBM T.J. Watson Research Center, Yorktown Heights, NY 10598, USA}
\author{Javier Robledo Moreno}\affiliation{
IBM Quantum, IBM T.J. Watson Research Center, Yorktown Heights, NY 10598, USA}
\author{Antonio Mezzacapo}
\thanks{Current affiliation: NVIDIA Corporation, Santa Clara, CA, USA}
\affiliation{
IBM Quantum, IBM T.J. Watson Research Center, Yorktown Heights, NY 10598, USA}
\author{Isaac L. Chuang}\affiliation{Department of Physics, Massachusetts Institute of Technology, Cambridge, Massachusetts 02139, USA}\affiliation{Center for Ultracold Atoms and Research Laboratory of Electronics, Massachusetts Institute of Technology, Cambridge, Massachusetts 02139, USA} 
\author{William Kirby}\affiliation{
IBM Quantum, IBM T.J. Watson Research Center, Yorktown Heights, NY 10598, USA}

\maketitle
\begin{abstract}
    Computing quantum ground states is generically difficult, but additional structure can sometimes allow diagonalization to be recast as a more feasible problem. For example, when the desired ground state is sparse in a given basis, diagonalization can be facilitated via graph search. We make this reformulation precise by introducing the \emph{eigenwalk problem}, which seeks the support of a sparse eigenvector of a Hermitian matrix by exploring the graph induced by its nonzero entries. However, it is not obvious whether the relevant support vertices must always be efficiently reachable by a search on the graph. To resolve this question, we prove that for every sparse eigenvector, there exists a (possibly different) sparse eigenvector with the same eigenvalue whose support is tightly localized in the graph, with diameter scaling only linearly in the sparsity and independently of the total number of vertices. As a consequence, if a $2^n$-dimensional, $\poly(n)$-sparse Hamiltonian has an $\mathcal{O}(1)$-sparse extremal eigenvector and one support element is known, then an exact eigenvector with the same eigenvalue can be computed classically in $\poly(n)$ time. The same conclusion follows when the $\mathcal{O}(1)$-sparse eigenvector is non-extremal, provided that it is sparser than every eigenvector with a different eigenvalue. These results hold with no assumptions on the degeneracy, locality, spectral width, or spectral gap of the Hamiltonian, and the underlying support-localization principle also extends to problems beyond exact diagonalization, such as sparse principal component analysis.

\end{abstract}

\section{Introduction}
Computing the ground state of quantum systems is a fundamental challenge in chemistry and physics~\cite{friesner2005ab,foulkes2001quantum,schollwock2005density,bartlett2007coupled,leblanc2015solutions,dirac1928quantum,kohn1999nobel,pople1999nobel}. In the context of chemistry, the ground state informs molecular geometry, charge distribution, and chemical stability; likewise, the ground state of condensed-matter systems determines low-temperature phase structure plus magnetic and topological properties. The lowest-energy states of a Hamiltonian can be constructed through exact dense diagonalization, but this procedure generically incurs exponential cost relative to the system size. Consequently, it has long been widely desirable to leverage additional mathematical structures to simplify diagonalization in certain settings.

Interestingly, when the desired ground state is known to be sparse in a given basis, its computation can be reformulated as a graph search problem. If the support of the ground state (i.e., the set of basis states with nonzero amplitudes) can be identified, then the Hamiltonian can be projected to the span of this support prior to diagonalization, often leading to substantial cost savings. One way to systematically search for the support is to harness the implicit graph structure of the underlying eigenvalue equation: note that the nonzero entries of the Hamiltonian matrix define edges that connect pairs of basis states, forming a graph. Ideally, the support vertices can be efficiently found by exploring this graph, thereby simplifying the ground state computation.

Analogous graph search problems also appear elsewhere in the field of computer science. Sparse principal component analysis, for instance, instead considers the graph induced by nonzero entries of a covariance matrix $\Sigma$ and seeks the support vertices of a $k$-sparse vector $v$ that maximizes the variance $v^\intercal \Sigma v$, where $k$ is a given positive integer \cite{spca}. Additionally, planted cliques can be found in a graph by locating vertices with large components in the second-highest eigenvector of the adjacency matrix \cite{planted-clique}. Along related lines, the identification of small subgraphs that locally support sparse eigenvectors of graph matrices (such as adjacency matrices and Laplacians) can simplify the computation of the remaining spectrum \cite{Henson}.

However, it is not evident whether sparse eigenvector supports can always be reached efficiently by graph search. For the exact ground-state problem, even if the desired support vertices are few, it seems plausible that they could be distributed arbitrarily far apart on the graph since the eigenvalue equation does not obviously constrain the support to be tightly localized. A priori, one might expect that reaching the desired vertices could require visiting an exponentially large portion of the graph, even when the support itself is small. 

At the same time, graph search has been demonstrated to be a practically useful subroutine in quantum chemistry. For example, a diverse class of empirical ground-state methods succeeds precisely because vertices which dominate the support tend to concentrate within an efficiently accessible region of the graph. Under standard assumptions, the Kaniel--Paige--Saad bounds imply that an accurate ground-state approximation lies in a Krylov subspace of dimension polynomial in the system size, generated by repeated applications of the Hamiltonian; if the initial vector is supported on a polynomial-diameter subgraph, then the support of this approximation must lie within a polynomial-diameter subgraph as well~\cite{kaniel1966estimates,paige1971computation,saad}. However, even when the diameter of the latter subgraph is polynomial, fully enumerating it could require visiting an exponential number of vertices.
Modern selected configuration interaction (CI)~\cite{booth2009fermion,tubman2016deterministic,holmes2016heat,schriber2016communication} narrows this large search space heuristically: in practice, one often needs only a polynomially-sized subset of the reached vertices to obtain an accurate ground-state approximation~\cite{bender1969studies,huron1973iterative,buenker1974individualized,buenker1978applicability,evangelisti1983convergence,illas1991selected,harrison1991approximating,stampfuss2005improved,ivanic2001identification,bytautas2009priori,epstein1926stark,nesbet1955configuration}.
Some recent varieties of selected CI have even been formulated explicitly in graph-theoretic language~\cite{Sun2023,trimci}. 
Nonetheless, there is no known proof that selected CI methods find ground-state approximations in polynomial time in general; therefore, they do not imply efficient identification of the exact ground-state support either, a strictly stronger task.

In recent years, the graph-search motif has also appeared in quantum algorithms for approximating sparse ground states: these techniques randomly sample support vertices by measuring qubits after the application of a quantum circuit~\cite{sqd,kanno2023quantum,mikkelsen2024quantum,sugisaki2025hamiltonian,skqd,parrish2019quantum,stair2020multireference} that effectively explores the graph. In fact, quantum Krylov methods based on steps of time evolution even offer rigorous guarantees of polynomial runtime for sparse ground-state approximation under certain assumptions~\cite{skqd,epperly2022theory}. In practice, convergence is often more efficient than predicted by linear-algebraic arguments~\cite{kirby2024analysis}, suggesting the intriguing possibility of a deeper interpretation, conceivably based on graph-theoretic ideas.

In this paper, we establish that a graph search for the support vertices of an exact eigenvector can be provably efficient on a classical computer. To do so, we introduce the \textit{eigenwalk problem}: given a simple graph on $2^n$ vertices and a Hermitian matrix whose nonzero entries correspond exactly to the edges of the graph, find the support of any exact eigenvector that belongs to a specified target eigenspace. This target eigenspace is identified by the index of its eigenvalue (e.g., lowest, second lowest, etc.), and the eigenvalue itself is not assumed to be known beforehand. We prove that when the matrix is $\poly(n)$-sparse and the target eigenspace contains an $\mathcal{O}(1)$-sparse eigenvector $\ket{\psi}$, then the desired support lies in a subgraph of \textit{constant} diameter that can be enumerated by visiting $\poly(n)$ vertices. The eigenwalk is therefore solvable in $\poly(n)$ time via a classical algorithm, assuming initial knowledge of any vertex in the support of $\ket{\psi}$ (an additional ``sparsity-separation" promise is assumed if the eigenvector is non-extremal). It ultimately follows that the ground state of a $\poly(n)$-sparse Hamiltonian can be computed exactly in polynomial time, provided the ground state is $\mathcal{O}(1)$-sparse and one has access to a $\poly(n)$-sparse guiding state with nonzero overlap on the ground state. As a caveat, we note that while this runtime is formally polynomial, the algorithm may not always be practical in realistic settings because the polynomial exponent scales with the ground state sparsity. Nonetheless, we highlight two important differences between our setting and that of conventional Krylov methods: firstly, we seek to compute the supports of \textit{exact} eigenvectors, rather than approximations to those eigenvectors; secondly, in contrast to standard convergence assumptions for Krylov methods, our results presuppose no conditions on any spectral properties of the Hamiltonian (such as the spectral width or gap) or even the magnitude of the guiding state overlap, provided it is nonzero. 

As intermediate steps, we also prove a number of interesting facts about the structure of eigenvector supports within graphs. For example, even when the support vertices of a given eigenvector are distributed arbitrarily far apart in the graph, we prove that there must exist some other eigenvector with the same eigenvalue whose support is instead tightly localized. In particular, the graph diameter of the latter's support vertices scales at most linearly with the size of the support and independently of the total size of the graph. We additionally demonstrate the robustness of this result in the case where the underlying eigenvector is only approximately sparse. While our primary case-study in this work is Hamiltonians appearing in quantum mechanics, our results have broader applicability, for example to sparse principal component analysis.

The rest of this paper is structured as follows. In Section~\ref{sec:graph-struct}, we establish these facts about the structure of eigenvector supports within graphs. Then, in Section~\ref{sec:algo}, we use these results to formulate a polynomial-time algorithm for the eigenwalk problem. Lastly, we discuss the significance of this algorithm and its implications for quantum advantage in Section~\ref{sec:discussion}. The Appendix then completes several proofs deferred from the main text, provides illustrative examples to supplement arguments made in the main text, and applies our results to sparse principal component analysis.

\section{Eigenvector supports within graphs}\label{sec:graph-struct}
If an eigenvector is to be supported on a particular subset of vertices, then the graph must satisfy corresponding structural constraints; conversely, the structure of that support within the graph may impose linear-algebraic constraints on the underlying eigenvector, such as degeneracy. After formalizing the eigenwalk problem in Section~\ref{sec:eigenwalk-statement}, we make these claims concrete in Sections~\ref{sec:graph-eigenvalue} and \ref{sec:irreducible}. 

\subsection{Eigenwalk problem statement}\label{sec:eigenwalk-statement}
The eigenwalk problem is defined as follows. Let $G$ be a simple graph with vertices $V$ and edges $E$, and consider the vector space spanned by an orthonormal basis of vectors $\ket{j}$ indexed by the vertices $j\in V$. A Hermitian operator $H$ on this vector space is said to be \textit{$G$-consistent} if $H_{ij} \neq 0$ for all edges $\{i,j\}\in E$ and $H_{ij} = 0$ for all pairs $i\neq j$ which are not edges. Given a $G$-consistent $H$ with $m$ unknown distinct eigenvalues $\lambda_1<\ldots<\lambda_m$ and a target eigenvalue index $k\in\{1,\ldots, m\}$, the eigenwalk problem asks for the vertex support of any eigenvector $\ket{\psi}$ of $H$ with eigenvalue $\lambda_k$ (see Figure~\ref{fig:eigenwalk-problem}). The vertex support of $\ket{\psi}=\sum_{j\in V}\psi_j \ket{j}$, denoted $\supp(\ket{\psi})$, is defined as the subset of vertices $j\in V$ with nonzero complex coefficients $\psi_j$. Note that the eigenvector $\ket{\psi}$ is said to be \textit{extremal} if its eigenvalue is $\lambda_1$ or $\lambda_m$, and \textit{non-extremal} otherwise. If $\ket{\psi}$ is nondegenerate, then the eigenwalk problem seeks a unique set of support vertices. On the other hand, if $\ket{\psi}$ is degenerate, the problem accepts multiple valid solutions, corresponding to the supports of all possible degenerate eigenvectors. The eigenwalk problem is furthermore said to be \textit{guided} if access is available to at least one ``guiding" vertex promised to belong to the (otherwise unknown) support of some satisfactory eigenvector.

\begin{figure}[htbp]
    \centering
    \includegraphics[width = 0.60\linewidth]{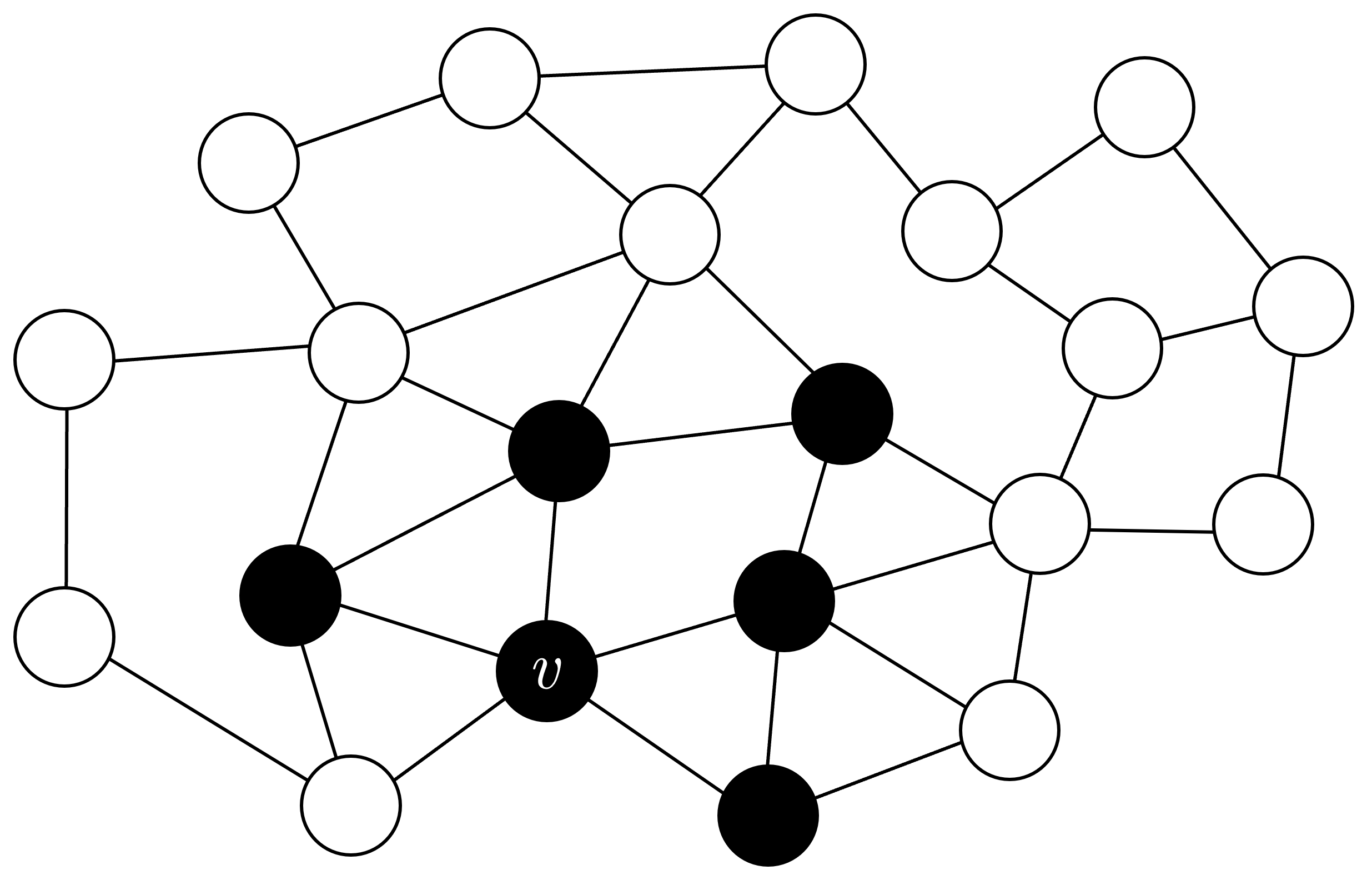}
    \caption{The eigenwalk problem for a Hermitian matrix $H$ consistent with the above graph: given a target eigenvalue index $k$, find the vertex support of any eigenvector of $H$ associated with the $k$th-smallest eigenvalue. The support of one such eigenvector $\ket{\psi}$ is the set of shaded vertices, which constitutes a valid solution. If a guiding vertex $v$ belonging to $\supp(\ket{\psi})$ is known, the remaining support vertices can be reached by traversing the graph outward from $v$.} 
    \label{fig:eigenwalk-problem}
\end{figure}

The guided eigenwalk problem could be solved naively by exactly diagonalizing the full matrix $H$ and reading off the support of the desired eigenvector; however, such an approach would in general incur exponential cost. On the other hand, this problem could be more efficiently solved if some solution set of support vertices were tightly localized in $G$. Supposing these vertices all belonged to a small connected subgraph, they could be reached quickly by searching outwards from the guiding vertex. 

Remarkably, the underlying structure of the eigenvalue equation does indeed force certain solution sets of support vertices to be tightly localized: in particular, one can always find a solution whose graph diameter depends linearly on the number of support vertices and \textit{not} on the total number of vertices $\abs{V}$. To intuitively understand this fact, we will build some simple pedagogical examples, while deferring certain rigorous proofs to Appendix \ref{appendix:proofs}.

\subsection{Graph structure and eigenspace structure} \label{sec:graph-eigenvalue}

First, we explore how the eigenvalue equation imposes constraints on the underlying graph. Given any graph $G$ and subset $S$ of vertices, there does not always exist a $G$-consistent Hermitian matrix $H$ that admits an eigenvector whose support is $S$. For example, simply let $G$ consist of two connected vertices as in Figure \ref{fig:graph-examples}a, and let $S$ contain only vertex 1. Suppose there exists some $G$-consistent Hermitian $H$ with an eigenvector $\ket{\psi} = \ket{1}$. Then $H\ket{\psi} = H_{11}\ket{1} + H_{21}\ket{2}$. However, since $H$ is $G$-consistent, $H_{21}\neq 0$, so $H\ket{\psi}$ cannot be proportional to $\ket{\psi}$, and $\ket{\psi}$ cannot be an eigenvector of $H$. The issue arises on the non-support vertices adjacent to $S$, where the amplitudes of $H\ket{\psi}$ must cancel. Accordingly, after modifying $G$ by adding a third vertex adjacent to the existing two and expanding $S$ to $\{1,3\}$ (see Figure \ref{fig:graph-examples}b), it is possible for some $G$-consistent $H$ to have an eigenvector $\ket{\psi} = c_1\ket{1}+c_3\ket{3}$ provided that the entries of $H$ and coefficients $c_1,c_3$ are chosen such that $\bra{2}H\ket{\psi}=H_{21}c_1+H_{23}c_3 = 0$. For example, it suffices to choose $H$ as the adjacency matrix of $G$, that is, where $H_{ij} = 1$ if $\{i,j\}$ is an edge of $G$ and 0 otherwise:
\begin{align} \label{eq:H-example}
    H = \begin{pmatrix}
        0 &1 &1\\
        1 &0 &1\\
        1 &1 &0
    \end{pmatrix}.
\end{align}
Then ${\ket{\psi} = \ket{1}-\ket{3}}$ is an eigenvector of $H$ with eigenvalue $-1$.
\begin{figure}[htbp]
    \centering
    \includegraphics[width = 0.55\linewidth]{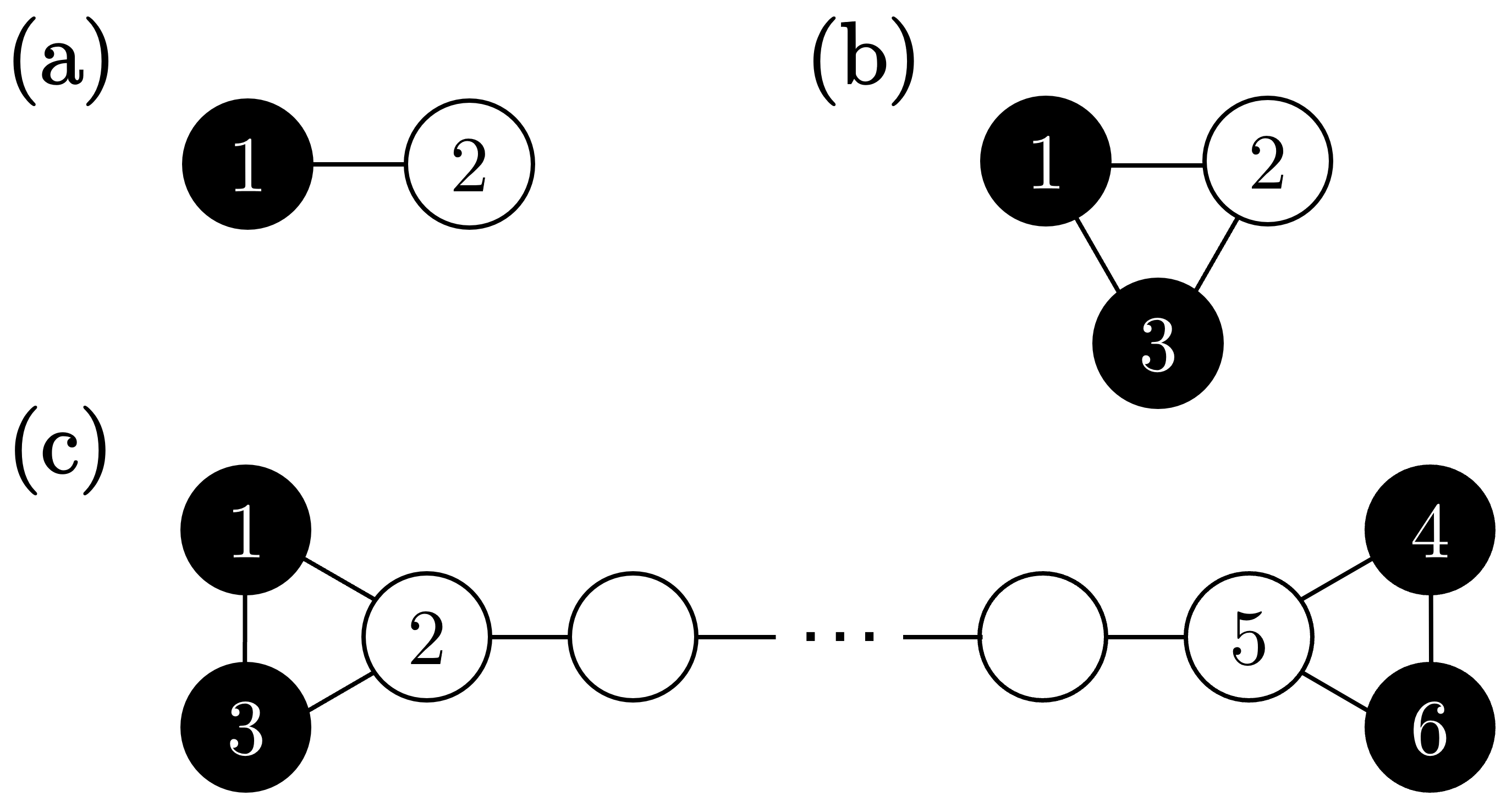}
    \caption{Three possible graphs, each with a subset $S$ of shaded vertices. In \textbf{(a)}, there exists no graph-consistent Hermitian matrix admitting an eigenvector whose support is $S$, while in \textbf{(b)} and \textbf{(c)}, there does exist such a matrix. In \textbf{(c)}, the resulting support $S$ can be partitioned into two vertex pairs which are separated by an arbitrarily large distance.} 
    \label{fig:graph-examples}
\end{figure}

Next, we examine whether there could exist \textit{some} solution to an eigenwalk problem whose vertices are distributed arbitrarily far apart in the underlying graph. Construct a graph $G'$ by connecting two copies of the $G$ from Figure \ref{fig:graph-examples}b through an arbitrarily long path of $\ell$ extra vertices, as depicted in Figure \ref{fig:graph-examples}c. Then define $H'$ again as the adjacency matrix of $G'$. It is straightforward to verify that ${\ket{\psi'} = (\ket{1}-\ket{3})+(\ket{4}-\ket{6})}$ is an eigenvector of $H'$ with eigenvalue $-1$. Moreover, the support vertices $\{1,3\}$ are separated from the remaining support $\{4,6\}$ by a distance of $\ell + 3$, which can be made arbitrarily large, as desired. 

However, when the support vertices are divided into sufficiently separated subsets, the underlying eigenvector must be degenerate. In the present example, the components of $\ket{\psi'}$ supported on each of $\{1,3\}$ and $\{4,6\}$, that is, $(\ket{1}-\ket{3})$ and $(\ket{4}-\ket{6})$, are in fact also eigenvectors of $H'$ with eigenvalue $-1$. In general, degeneracy is guaranteed when the full eigenvector support can be partitioned into subsets with disjoint neighborhoods in the graph (see Figure~\ref{fig:venn}). Interestingly, for extremal eigenvectors, an easier sufficient condition applies: the support need only admit a partition with no edges between distinct subsets, even if their neighborhoods overlap.

\begin{figure}[htbp]
    \centering
    \includegraphics[width = 0.7\linewidth]{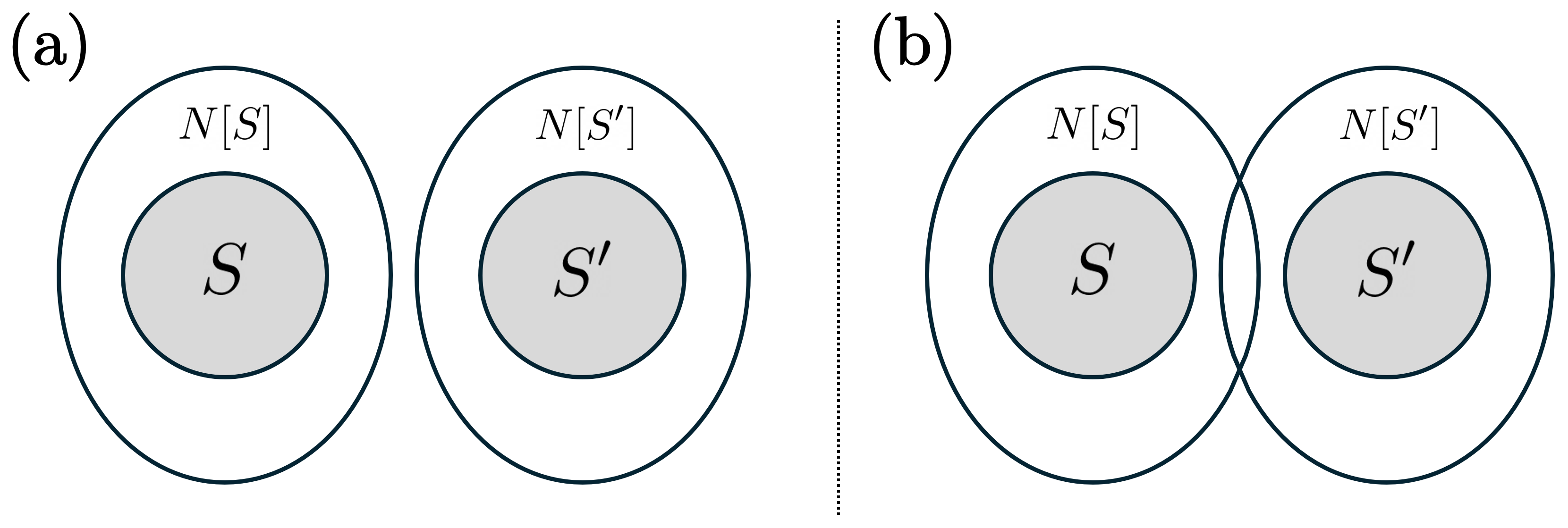}
    \caption{Two sufficient conditions for eigenvector degeneracy when the support is partitioned into subsets $S$ and $S'$. \textbf{(a)} In general, if $S$ and $S'$ have disjoint closed neighborhoods in the graph, the eigenvector must be degenerate. \textbf{(b)} In the extremal case, degeneracy follows if no vertices in $S$ are neighbors of $S'$, and vice versa.}
    \label{fig:venn}
\end{figure}

The following theorem summarizes these results. For any graph $G$ and two subsets $S,S'$ of vertices, define $\operatorname{dist}(S,S')$ to be the minimum distance (number of edges) in $G$ between a vertex in $S$ and one in $S'$. Additionally, let $N[S]$ represent the closed neighborhood of $S$ in $G$, that is, the union of $S$ with any vertices adjacent to one in $S$; let $N(S) = N[S]\setminus S$ be the open neighborhood. Lastly, denote the projector onto the span of $S$ with ${\Pi_S = \sum_{j\in S} \ketbra{j}}$.

\begin{theorem}[Sufficient conditions for degeneracy]\label{thm:degeneracy}
    Let $G$ be a simple graph, and let $H$ be a $G$-consistent Hermitian matrix. Suppose that ${H\ket{\psi} =\lambda\ket{\psi}}$ for some eigenvector $\ket{\psi}$ and eigenvalue $\lambda$, and that $\supp(\ket{\psi})$ can be partitioned into two nonempty subsets $S$ and $S'$. If
    \begin{align}\label{eq:degeneracy-distance}
        \operatorname{dist}(S,S') > 
        \begin{cases}
        2, & \textrm{in general},\\
        1, & \textrm{when } \ket{\psi} \textrm{ is extremal,}
        \end{cases}
    \end{align}
    then there exist nonzero eigenvectors $\ket{\psi_S}$ and $\ket{\psi_{S'}}$ of $H$, both with eigenvalue $\lambda$, such that
    \begin{align}
        \supp(\ket{\psi_S})= S,
        \qquad
        \supp(\ket{\psi_{S'}})= S'.
    \end{align}
    Consequently, $\lambda$ is degenerate.
\end{theorem}
\begin{proof}
    General case. Let $ R =\supp(\ket{\psi})$, and suppose $ R $ can be partitioned into nonempty $S,S'$ with $\operatorname{dist}(S,S')>2$. Then $N[S] \cap N[S'] = \emptyset$. Since no matrix elements of $H$ connect vertices in $ R $ to vertices outside of $N[ R ]$, ${H\ket{\psi} = \Pi_{N[ R ]}H\Pi_ R \ket{\psi}}$. Furthermore, since $N[S]\cap N[S'] = \emptyset$, $N[ R ]$ can be partitioned into $S, N(S), S',$ and $N(S')$. Thus, defining $\ket{\psi_S} = \Pi_S\ket{\psi}$ and $\ket{\psi_{S'}}= \Pi_{S'}\ket{\psi}$,
    \begin{align}\label{eq:gen-degen-01}
        H\ket{\psi} &= \left(\Pi_S + \Pi_{N(S)}+\Pi_{S'}+\Pi_{N(S')}\right)H\left(\Pi_S+\Pi_{S'}\right)\ket{\psi}\nonumber\\
        &= \big(\Pi_S + \Pi_{N(S)}\big)H\ket{\psi_S}+\big(\Pi_{S'}+\Pi_{N(S')}\big)H\ket{\psi_{S'}}.
    \end{align}
    Note that $\Pi_S,\Pi_{N(S)},\Pi_{S'},$ and $\Pi_{N(S')}$ are mutually orthogonal projectors. Therefore, matching terms with $H\ket{\psi} = \lambda\ket{\psi_S}+\lambda\ket{\psi_{S'}}$ yields
    \begin{gather}
        \Pi_SH\ket{\psi_S} = \lambda\ket{\psi_S},\ \Pi_{S'}H\ket{\psi_{S'}} = \lambda\ket{\psi_{S'}},\ \textrm{and}\nonumber\\
        \Pi_{N(S)}H\ket{\psi_S} = \Pi_{N(S')}H\ket{\psi_{S'}} = 0.
    \end{gather}
    Consequently, $\ket{\psi_S}$ is an eigenvector of $H$ with eigenvalue $\lambda$, since
    \begin{align}
        H\ket{\psi_S} = (\Pi_S + \Pi_{N(S)})H\ket{\psi_S} = \lambda\ket{\psi_S}.
    \end{align}
    An identical argument follows for $\ket{\psi_{S'}}$.

    Extremal case. Assume without loss of generality that $\ket{\psi}$ is a normalized lowest eigenvector, and suppose $ R $ can be partitioned into nonempty $S,S'$ with $\operatorname{dist}(S,S')>1$. Since $S,S'$ are disjoint, $\ket{\psi}$ can be written as ${\ket{\psi} = \ket{\psi_S}+\ket{\psi_{S'}}}$, where $\ket{\psi_S} = \Pi_S\ket{\psi}$ and $\ket{\psi_{S'}}=\Pi_{S'}\ket{\psi}$ are orthogonal. Since $S$ and $S'$ are separated by a minimum distance of $>1$, no vertices in $S$ are neighbors of $S'$, and vice versa. Hence, no nonzero elements of $H$ connect vertices in $S$ to those in $S'$, so $\bra{\psi_{S'}}H\ket{\psi_{S}}=0$. It follows that $\bra{\psi}H\ket{\psi} = \bra{\psi_S}H\ket{\psi_S} + \bra{\psi_{S'}}H\ket{\psi_{S'}}$. Now define the normalized state $\ket{\psi'} = \ket{\psi_S}-\ket{\psi_{S'}}$. Since $\bra{\psi'}H\ket{\psi'} = \bra{\psi_S}H\ket{\psi_S} + \bra{\psi_{S'}}H\ket{\psi_{S'}}$ as well, the Rayleigh--Ritz theorem (see Appendix \ref{appendix:rayleigh}) implies that $\ket{\psi'}$ must also be a lowest eigenvector. Lastly, since $\ket{\psi_S}$ and $\ket{\psi_{S'}}$ are in $\operatorname{Span}(\ket{\psi},\ket{\psi'})$, they are lowest eigenvectors too, as desired.
\end{proof}

In Appendix \ref{appendix:examples12}, we demonstrate that the sufficient conditions in Eq. (\ref{eq:degeneracy-distance}) are tight. Furthermore, we show that Theorem~\ref{thm:degeneracy} has no converse necessary condition: in particular, given two degenerate eigenvectors whose supports are disjoint sets, these sets need not be separated by a distance $>1$ in the graph.

The Rayleigh--Ritz theorem is the crucial resource which distinguishes the extremal case from the general one. In the extremal proof above, a state $\ket{\psi'}$ is constructed which shares the same energy expectation value as $\ket{\psi}$. Because this value is an extremal eigenvalue of $H$, Rayleigh--Ritz implies that $\ket{\psi'}$ is also an eigenvector with the same eigenvalue, as desired. This reasoning fails for non-extremal $\ket{\psi}$, since a state may have the same energy expectation value without being an eigenvector.

Although there exist certain solutions to eigenwalk problems whose vertices are distributed arbitrarily far apart in the graph, Theorem~\ref{thm:degeneracy} shows that such examples can be broken into smaller equivalent solutions which are more localized. In particular, if the original eigenvector's support could be divided into disjoint subsets which are sufficiently separated, these subsets would support degenerate eigenvectors and therefore also be admissible eigenwalk solutions. Given an eigenvector $\ket{\psi}$, define its \textit{irreducible support components} to be the parts of the finest partition of $\supp(\ket{\psi})$ such that any two distinct parts have graph distance greater than $\tau$, where $\tau=1$ if $\ket{\psi}$ is extremal and $\tau=2$ otherwise. Note that irreducibility is only with respect to the separation criterion of Theorem~\ref{thm:degeneracy}: although each irreducible support component supports an eigenvector degenerate with $\ket{\psi}$, smaller subsets may still in principle support eigenvectors with the same eigenvalue.

As a brief detour, the notion of irreducible support components enables a more rigorous description of how the eigenvalue equation may restrict graph structure, per the following theorem:
\vspace{20pt}
\begin{theorem}[Eigenvalue equation restricts graph structure]\label{thm:eigstruct}
    Let $G=(V,E)$ be a simple graph and let $S\subseteq V$. If there exists a $G$-consistent Hermitian matrix $H$ admitting an eigenvector $\ket{\psi}$ with $\supp(\ket{\psi})=S$, then for every irreducible support component $T$ of $\ket{\psi}$ and every non-support vertex $v\in V\setminus S$, $v$ cannot be adjacent to exactly one vertex of $T$.
\end{theorem}
\begin{proof}
    Suppose there exists such an $H$ and $\ket{\psi}$ with eigenvalue $\lambda$, and assume that some vertex $v\in V\setminus S$ is adjacent to exactly one vertex $v'\in T$, where $T$ is an irreducible support component of $\ket{\psi}$. Let $\ket{\psi_T} = \Pi_T\ket{\psi}$. Then 
    \begin{align}
        H\ket{\psi_T} = \Pi_V H\ket{\psi_T} =\Pi_{V\setminus v}H\ket{\psi_T} + H_{vv'}c_{v'}\ket{v},
    \end{align}
    where $c_{v'} = \braket{v'|\psi_T}$ is nonzero since $v'\in T$. Since $v$ and $v'$ are adjacent, $H_{vv'}\neq 0$ too, so $\bra{v}H\ket{\psi_T} \neq 0$. At the same time, Theorem~\ref{thm:degeneracy} implies that $\ket{\psi_T}$ is an eigenvector degenerate with $\ket{\psi}$, so $\bra{v}H\ket{\psi_T} = \lambda \braket{v|\psi_T}$ as well. However, since $v\notin T$, $\ket{v}$ and $\ket{\psi_T}$ must be orthogonal, so $\bra{v}H\ket{\psi_T}$ must be zero, leading to a contradiction. 
 \end{proof}

Theorem~\ref{thm:eigstruct} reveals an interesting graph-structural phenomenon that can occur only in the non-extremal setting. Consider a three-vertex path graph $G$ and let $S$ be the set of two endpoint vertices. It is possible to construct a $G$-consistent $H$ with a nondegenerate non-extremal eigenvector $\ket{\psi}$ supported on $S$ (see Appendix \ref{appendix:examples12})---this possibility is consistent with the above theorem, since $S$ itself is the sole irreducible support component, and the middle non-support vertex is adjacent to two vertices in $S$. The support of $\ket{\psi}$ is only connected through the middle ``virtual" vertex outside the support, yet it is impossible to decompose $\ket{\psi}$ into sparser degenerate eigenvectors (since $\ket{\psi}$ is constructed to be nondegenerate). Intriguingly, such a virtual vertex cannot exist in the extremal setting: if $\ket{\psi}$ were extremal, each support vertex would be a distinct irreducible support component, and the middle vertex would be adjacent to exactly one vertex per component. Due to Theorem~\ref{thm:eigstruct}, it would then be impossible to find \textit{any} extremal eigenvector supported on $S$, even if it were allowed to be degenerate.

\begin{figure}[htbp]
    \centering
    \includegraphics[width = 0.7\linewidth]{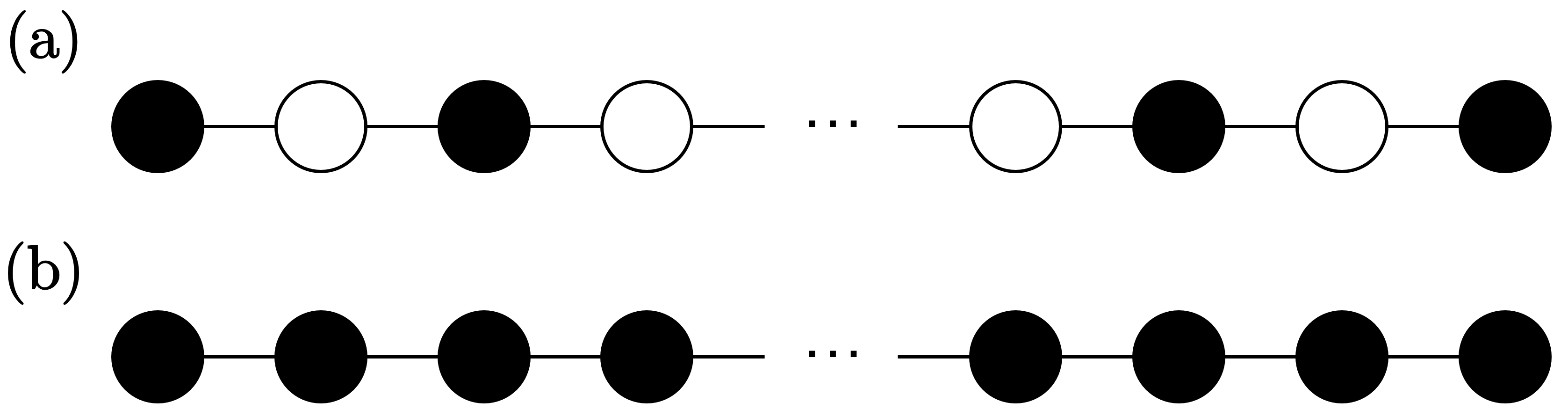}
    \caption{Example graphs saturating diameter bound for irreducible support components. Let the subset $T$ of shaded vertices be an irreducible support component of some eigenvector $\ket{\psi}$. If $\ket{\psi}$ is non-extremal, the alternating path graph in \textbf{(a)} realizes the maximum possible diameter $2|T|-2$. If instead $\ket{\psi}$ is extremal, the path graph in \textbf{(b)} realizes the maximum possible diameter $|T|-1$.} 
    \label{fig:paths}
\end{figure}

\subsection{Localization of irreducible support components}\label{sec:irreducible}
Returning to the topic of eigenwalks, recall that irreducible support components provide equivalent solutions to an eigenwalk problem; we now show that these components must be tightly localized on the graph. The intuition is simple: if an irreducible support component containing a fixed number of vertices had too large a diameter (i.e., maximum distance between any two of its vertices), then it would admit a further partition into subsets separated beyond the relevant distance threshold, contradicting irreducibility. Let $T$ be an irreducible support component of some eigenvector, and let $\tau = 2$ when the eigenvector is non-extremal and $\tau=1$ when it is extremal. Since $T$ is irreducible, it cannot be split into subsets separated by a distance greater than $\tau$. Consequently, any two vertices in $T$ must be connected by some simple path where consecutive members of $T$ along the path are separated by at most $\tau$ edges (see Figure \ref{fig:paths}). Such a path visits at most $\abs{T}$ vertices of $T$, so there are at most $\abs{T}-1$ pairs of successive vertices from $T$. Since each of these pairs is separated by a distance of at most $\tau$, the total path length is at most $\tau(\abs{T}-1)$. Hence, the diameter of $T$ is also at most $\tau(\abs{T}-1)$. The theorem below, which follows from Theorem~\ref{thm:degeneracy}, states this result more rigorously.

\begin{theorem}[Localization of irreducible support components]\label{thm:localized}
    Let $G$ be a simple graph, and let $H$ be a $G$-consistent Hermitian matrix with an eigenvector $\ket{\psi}$ corresponding to eigenvalue $\lambda$. Set $s=\abs{\supp(\ket{\psi})}$. Then for every vertex $v\in\supp(\ket{\psi})$, there exists an eigenvector $\ket{\psi_v}$ of $H$ with eigenvalue $\lambda$ such that $v\in\supp(\ket{\psi_v})$ and
    \begin{align}
        \operatorname{diam}_G\bigl(\supp(\ket{\psi_v})\bigr) \leq
        \begin{cases}
            2s-2, & \text{if } \ket{\psi} \text{ is non-extremal},\\[4pt]
            s-1, & \text{if } \ket{\psi} \text{ is extremal},
        \end{cases}
    \end{align}
    where $\operatorname{diam}_G(T):=\max_{i,j\in T}\operatorname{dist}(i,j)$ denotes the graph diameter of $T$ in $G$.
\end{theorem}
\begin{proof}
    See Appendix \ref{appendix:general}.
\end{proof}
Note that the above upper bound depends only on the size of $\supp(\ket{\psi})$ and is independent of the total number of vertices $\abs{V}$ in the graph. Additionally, this bound is tight (see Appendix \ref{appendix:examples3}).

The ground states involved in realistic settings are often only approximately sparse; that is, most of the eigenvector's amplitude is concentrated on a small support space, with some residual component outside the space. As an aside, we now show that the result of Theorem~\ref{thm:localized} is robust to approximate sparsity, at least in the extremal case. Informally, we show that if an extremal eigenvector $\ket{\psi}$ is approximately sparse, then there exists a sparse state $\ket{\psi'}$ that is close in energy to $\ket{\psi}$ and whose support is tightly localized on $G$.

Before formalizing this robustness result in the theorem below, it is necessary to precisely define approximate sparsity. Given a graph $G=(V,E)$, a vector $\ket{\psi} = \sum_{j\in V} \psi_j \ket{j}$ is said to be $(\eta,\epsilon)$-sparse over the vertex set $ A \subseteq V$ if $\abs{\psi_j}^2\geq \eta$ for all $j\in  A $ and $\sum_{j\in V\setminus  A } \abs{\psi_j}^2\leq \epsilon$ for some nonnegative parameters $\eta<1,\epsilon<\frac{9}{25}$. Intuitively, the parameter $\eta$ bounds the \textit{local} weights on each of the approximate-support vertices, while $\epsilon$ bounds the \textit{global} weight lying outside the approximate support. 

\begin{theorem}[Approximately sparse support localization]\label{thm:approx}
    Let $G$ be a simple graph, and let $H$ be a $G$-consistent Hermitian matrix with a normalized extremal eigenvector $\ket{\psi}$ corresponding to eigenvalue $\lambda$. Suppose $\ket{\psi}$ is $(\eta,\epsilon)$-sparse over the vertex set $ A \subseteq V$, and let $s = \abs{ A }$. Then, given any vertex $v\in  A $, there exists an exactly sparse normalized state $\ket{\psi_v}$ with $v\in\supp(\ket{\psi_v})$ and $\supp(\ket{\psi_v})\subseteq  A $ such that
    \begin{align}\label{eq:approx-cluster}
        \operatorname{diam}_G\bigl(\supp(\ket{\psi_v})\bigr) \leq s-1
    \end{align}
    and 
    \begin{align}\label{eq:approx-energy}
        \abs{\bra{\psi_v}H\ket{\psi_v}-\lambda} \leq \frac{2\norm{H}}{\eta}\sqrt{\epsilon} + \mathcal{O}(\epsilon).
    \end{align}
\end{theorem}
\begin{proof}
    See Appendix~\ref{appendix:approx} and Figure~\ref{fig:thm4proof}.
\end{proof}
We note that for the bound in Eq. (\ref{eq:approx-energy}) to not diverge as $\eta,\epsilon\rightarrow 0$, it is necessary for $\eta$ to scale as $\Omega(\sqrt{\epsilon})$. This assumption is not always reasonable in physical systems. However, by invoking an additional natural assumption about the guiding vertex $v$, the leading term of Eq. (\ref{eq:approx-energy}) can be modified to scale only with $\sqrt{\epsilon}$, with no explicit or problematic dependence on $\eta$ (see Appendix \ref{appendix:approx}).

\section{Solution to the eigenwalk problem}\label{sec:algo}
We now describe a polynomial-time algorithm for solving guided eigenwalk problems when the underlying eigenvector is exactly sparse. The theorem below, which follows from Theorem~\ref{thm:localized}, summarizes the main result. Note that for any positive integer $s$, a state is said to be $s$-sparse if its support has size $\leq s$.
\begin{theorem}[Polynomial-time algorithm for the sparse guided eigenwalk problem]\label{thm:alg}
    Suppose the following are given:
    \begin{itemize}
        \item A simple graph $G=(V,E)$ of maximum degree $d$,
        \item An integer $s$ satisfying $2\leq s\leq |V|$ (eigenvector sparsity),
        \item Positive integers $m\leq |V|$ (number of distinct eigenvalues) and $k\leq m$ (target eigenvalue index),
        \item Sparse row access to a $G$-consistent Hermitian matrix $H$ with unknown distinct eigenvalues ${\lambda_1<\ldots<\lambda_m}$, promised to have an unknown {$s\textrm{-sparse}$} eigenvector $\ket{\psi}$ with eigenvalue $\lambda_k$,
        \item A guiding vertex $v$ promised to belong to $\supp(\ket{\psi})$.
    \end{itemize}
    If $k\in\{1,m\}$ (i.e., $k$ is extremal), then there exists a classical algorithm that returns the support of an eigenvector with eigenvalue $\lambda_k$ in time
    \begin{align}
        \mathcal{O}(d^{3s-3}).
    \end{align}
    Otherwise, if $k\notin\{1,m\}$ (i.e., $k$ is non-extremal) and every eigenvector with eigenvalue not equal to $\lambda_k$ has support size strictly larger than $s$, then such an algorithm exists that runs in time
    \begin{align}
        \mathcal{O}(d^{(2s-2)(s+3)}).
    \end{align}
\end{theorem}
\begin{proof}
    See Appendix \ref{appendix:alg}.
\end{proof}
In particular, when $G$ contains $2^n$ vertices, the maximum degree $d$ of $G$ is $\poly(n)$, and the eigenvector sparsity $s$ is $\mathcal{O}(1)$, the algorithm runs in $\poly(n)$ time. The requirement $d= \mathcal{O}(\poly (n))$ is equivalent to the requirement that $H$ is sparse, with $\mathcal{O}(\poly (n))$ nonzero entries per column. Note that $H$ need not be a local Hamiltonian; it need only be sparse. We remark that the algorithm of Theorem \ref{thm:alg} not only computes the desired eigenvector support, but also computes the eigenvalue and the coefficients of the eigenvector.

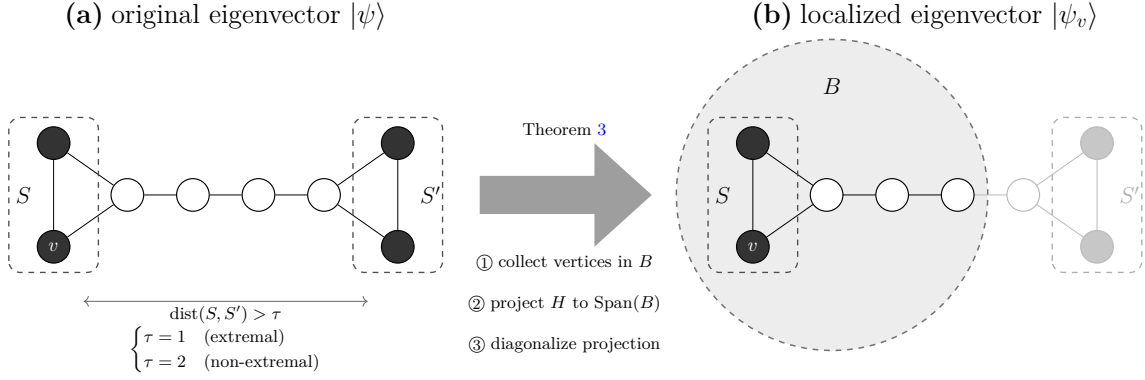
\begin{figure}[htbp]
\centering
\resizebox{\columnwidth}{!}{%
\begin{tikzpicture}[
  every node/.style={circle, draw=black, fill=white, inner sep=0pt,
                     minimum size=17pt, line width=0.6pt},
  support/.style   = {fill=black!80},
  guiding/.style   = {fill=black!80, text=white},
  collected/.style = {fill=white},
  faded/.style     = {draw=black!27, fill=white},
  fsupport/.style  = {draw=black!27, fill=black!22},
  lbl/.style  = {draw=none, fill=none, rectangle, font=\small,
                 inner sep=2pt, minimum size=0pt},
  lbl2/.style = {draw=none, fill=none, rectangle, font=\small,
                 inner sep=2pt, minimum size=0pt},
  lbl3/.style = {draw=none, fill=none, rectangle, font=\large,
                 inner sep=2pt, minimum size=0pt},
  lbltitle/.style = {draw=none, fill=none, rectangle, font=\Large,
                     inner sep=2pt, minimum size=0pt},
  fedge/.style = {draw=black!27, line width=0.50pt},
]


\node[support] (s1) at (0.00,  0.95) {};
\node[guiding] (v)  at (0.00, -0.95) {$v$};
\node          (l)  at (1.35,  0.00) {};

\node (p1) at (2.55, 0.00) {};
\node (p2) at (3.75, 0.00) {};

\node          (r)   at (4.95,  0.00) {};
\node[support] (sp1) at (6.30,  0.95) {};
\node[support] (sp2) at (6.30, -0.95) {};

\draw (s1)--(v);
\draw (v)--(l);
\draw (l)--(s1);

\draw (l)--(p1);
\draw (p1)--(p2);
\draw (p2)--(r);

\draw (r)--(sp1);
\draw (r)--(sp2);
\draw (sp1)--(sp2);

\draw[dashed, rounded corners=6pt, draw=black!70, line width=0.7pt]
  (-0.82,-1.42) rectangle (0.82,1.42);

\draw[dashed, rounded corners=6pt, draw=black!70, line width=0.7pt]
  (5.48,-1.42) rectangle (7.12,1.42);

\node[lbl3, anchor=east] at (-0.33, 0.00) {$S$};
\node[lbl3, anchor=west] at ( 6.63, 0.00) {$S'$};

\draw[<->, draw=black!60, line width=0.5pt]
  (0.55,-1.95) -- (5.75,-1.95)
  node[midway, below, lbl2, align=center]
  {$\operatorname{dist}(S,S')>\tau$\\[-1pt]
  $\begin{cases}
  \tau=1 & (\text{extremal})\\
  \tau=2 & (\text{non\mbox{-}extremal})
  \end{cases}$};

\node[lbltitle, anchor=south] at (3.15, 2.95)
  {\textbf{(a)}\;original eigenvector $\ket{\psi}$};


\fill[black!38]
  (7.80,-0.35) -- (9.90,-0.35)
  -- (9.90,-0.95)
  -- (10.95,0.00)
  -- (9.90,0.95)
  -- (9.90,0.35)
  -- (7.80,0.35)
  -- cycle;

\node[lbl2] at (9.35, 1.20) {Theorem~\ref{thm:localized}};

\node[lbl2, align=left] at (9.35, -1.25)
  {\textcircled{\raisebox{-0.8pt}{\footnotesize 1}}\;collect vertices in $B$};
\node[lbl2, align=left] at (9.35, -2.00)
  {\textcircled{\raisebox{-0.8pt}{\footnotesize 2}}\;project $H$ to $\operatorname{Span}(B)$};
\node[lbl2, align=left] at (9.35, -2.75)
  {\textcircled{\raisebox{-0.8pt}{\footnotesize 3}}\;diagonalize projection};
\begin{scope}[xshift=12.8cm]

\draw[dashed, draw=black!55, line width=0.8pt, fill=black!7]
  (1.45,0.00) circle (2.85cm);

\draw[fedge] (3.75, 0.00) -- (4.95, 0.00);
\draw[fedge] (4.95, 0.00) -- (6.30, 0.95);
\draw[fedge] (4.95, 0.00) -- (6.30,-0.95);
\draw[fedge] (6.30, 0.95) -- (6.30,-0.95);

\draw (0.00,  0.95) -- (0.00, -0.95);
\draw (0.00, -0.95) -- (1.35,  0.00);
\draw (1.35,  0.00) -- (0.00,  0.95);
\draw (1.35,  0.00) -- (2.55,  0.00);
\draw (2.55,  0.00) -- (3.75,  0.00);

\node[faded]    at (4.95, 0.00) {};
\node[fsupport] at (6.30, 0.95) {};
\node[fsupport] at (6.30,-0.95) {};

\node[support]   at (0.00,  0.95) {};
\node[guiding]   (vb) at (0.00, -0.95) {$v$};
\node[collected] at (1.35,  0.00) {};
\node[collected] at (2.55,  0.00) {};
\node[collected] at (3.75,  0.00) {};

\draw[dashed, rounded corners=6pt, draw=black!70, line width=0.7pt]
  (-0.82,-1.42) rectangle (0.82,1.42);

\draw[dashed, rounded corners=6pt, draw=black!32, line width=0.7pt]
  (5.48,-1.42) rectangle (7.12,1.42);

\node[lbl3, anchor=east] at (-0.33, 0.00) {$S$};
\node[lbl3, anchor=west, text=black!32] at (6.63, 0.00) {$S'$};
\node[lbl3] at (1.45, 2.00) {$B$};

\node[lbltitle, anchor=south] at (3.15, 2.95)
  {\textbf{(b)}\;localized eigenvector $\ket{\psi_v}$};

\end{scope}

\end{tikzpicture}%
}

\caption{Illustration of the polynomial-time algorithm for the guided sparse
eigenwalk problem (Theorem~\ref{thm:alg}).
\textbf{(a)}~The $s$ support vertices of a sparse eigenvector $\ket{\psi}$ may be split between two
subsets $S$ (containing the guiding vertex $v$) and $S'$ that are separated by more than $\tau$ edges. 
\textbf{(b)}~The same graph, with the ball $B$ of graph radius $s-1$ (extremal) or $2s-2$ (non-extremal) around $v$
shaded; all vertices outside $B$ are grayed.
By Theorem~\ref{thm:localized}, an eigenvector $\ket{\psi_v}$
with the same eigenvalue must exist whose support lies entirely inside
$B$.
The classical algorithm therefore
\textcircled{\raisebox{-0.8pt}{\footnotesize 1}}~collects the vertices in $B$,
\textcircled{\raisebox{-0.8pt}{\footnotesize 2}}~projects $H$ onto their span,
and
\textcircled{\raisebox{-0.8pt}{\footnotesize 3}}~exactly diagonalizes the projection,
yielding a solution in polynomial time.
The grayed set $S'$ and vertices outside $B$ need not be visited.}
\label{fig:algorithm}
\end{figure}

At a broad level, the algorithm leverages Theorem \ref{thm:localized}, which implies that there always exists a suitable eigenvector whose support is localized near the guiding vertex $v$ (see Figure~\ref{fig:algorithm}). Starting at $v$, the algorithm collects all vertices within distance $r$ of $v$ (where $r=s-1$ in the extremal case and $r = 2s-2$ otherwise), thereby ensuring that this localized support is contained in the collected set. Then, $H$ is projected to the span of these collected vertices before being exactly diagonalized. Any eigenvector of the full $H$ whose support lies in this span is also an eigenvector of the projection of $H$; hence, one of the eigenvectors obtained by diagonalizing the projection must correspond to a satisfactory solution. The algorithm identifies this solution eigenvector and simply returns its support. 

Between the extremal and non-extremal cases, the procedures for determining the solution eigenvector differ. The identification is simple in the extremal case, since the Rayleigh--Ritz theorem (see Appendix \ref{appendix:rayleigh}) implies that the lowest (highest) eigenvector of the projection is also a lowest (highest) eigenvector of the full Hamiltonian. However, locating the solution is more tricky in the non-extremal case, where the Rayleigh--Ritz theorem no longer applies. Recall that since the specific eigenvalue of the desired eigenvector is not known \textit{a priori}, the algorithm must instead identify this eigenvector by its eigenvalue index. For non-extremal eigenvectors, however, the projection step could obscure the eigenvalue ordering: for instance, a second excited state of $H$ could become a first excited state of the projected operator. To circumvent this issue, we demand an additional ``sparsity-separation" promise exclusively in the non-extremal setting, where every eigenvector outside the target eigenspace is guaranteed to have larger support than the desired eigenvector. Then, the solution can be identified from the eigenvectors of the projected Hamiltonian by filtering candidates according to sparsity.

As an immediate corollary, Theorem \ref{thm:alg} provides a polynomial-time algorithm for computing a ground state of a $2^n$-dimensional Hamiltonian $H$, provided that $H$ has $\poly(n)$ nonzero entries per column and one has access to a $\poly(n)$-sparse guiding state with nonzero overlap on a ground state of sparsity at most a known constant $s=\mathcal{O}(1)$. Since nonzero overlap implies that at least one vertex in the support of the
guiding state lies in the support of this ground state, the eigenwalk subroutine can be run from each vertex in the guiding support, and the
lowest-energy output can be selected to obtain a ground state.

If only an approximate ground state is required, the eigenwalk subroutine can be accelerated by replacing the exact diagonalization step with a Krylov algorithm such as the Lanczos method. In this case, $L$ Lanczos iterations cost approximately $\mathcal{O}(Ld^s)$, compared to the $\mathcal{O}(d^{3s-3})$ cost of exact diagonalization. The Kaniel--Paige--Saad bounds~\cite{kaniel1966estimates,paige1971computation,saad} imply that $L$ is polynomial under assumptions of polynomial spectral width, inverse-polynomial spectral gap, inverse-polynomial overlap between the guiding state and ground eigenspace, and inverse-polynomial additive precision. Thus, the Lanczos method can afford a speedup in practice when an approximation to the ground state is sufficient. However, when the \textit{exact} ground state support is desired (i.e., to explicitly solve the eigenwalk problem), the Lanczos method may fail: support is not stable under small perturbations, so recovering the
true support from an approximate eigenvector is not guaranteed.

\section{Discussion and conclusion} \label{sec:discussion}
In this paper, we introduced the eigenwalk problem and provided an efficient algorithm to solve it in the guided sparse setting. Even though we have primarily studied the eigenwalk problem in the context of computing quantum ground states, we note that the results established in this paper also extend to problems such as sparse principal component analysis, as discussed in Appendix~\ref{appendix:spca}. Nonetheless, our solution to the eigenwalk problem yields a polynomial-time algorithm for computing an $\mathcal{O}(1)$-sparse eigenvector of a $\poly(n)$-sparse Hamiltonian and recovering the eigenvector's exact support, given one guiding vertex in that support. 

Although this algorithm is proven to run in polynomial time, it may not be a replacement for existing ground-state heuristics. The polynomial exponent scales linearly with the ground-state sparsity, so the algorithm could rapidly become impractical for even modestly large supports. Furthermore, our runtime analysis only applies if the ground state is promised to be exactly sparse. In realistic scenarios, ground states are usually only approximately sparse, and heuristics may remain more effective.

With these caveats out of the way, what is unique about our algorithm is that it guarantees polynomial-time exact computation of sparse eigenvectors even when nothing is known about the spectrum of the Hamiltonian \textit{a priori}. In contrast, alternative methods for similar sparse eigenproblems are either heuristic or require that specific spectral criteria be met. For example, the truncated power method \cite{truncated} can approximate $\poly(n)$-sparse ground states in polynomial time, but this performance is only guaranteed under certain assumptions of nondegeneracy, spectral width and gap, and guiding-state overlap. Our eigenwalk explores only the unweighted connectivity graph of $H$, which encodes the positions of nonzero matrix entries but not their numerical values. Since many Hamiltonians with widely different spectra are consistent with the same graph, this graph does not preserve the same spectral information used by power methods. Furthermore, the localization of support vertices leveraged by our algorithm is a generic consequence of the eigenvalue-equation structure, independent of the particular eigenvalues involved. Our algorithm therefore requires no assumptions on degeneracy, spectral width or gap, or even the magnitude of the guiding-state overlap (provided that it is at least nonzero).

Our results also have implications for possible future quantum algorithms that could target the eigenwalk problem. A quantum algorithm genuinely beating our classical method in the setting of this paper would have to look entirely unlike the classical model: it would need an efficient coherent oracle for support identification, which is a qualitatively hard open problem, not a straightforward application of known quantum techniques. In particular, the bottleneck of our classical method is not graph traversal, but rather the subsequent diagonalization and post-processing steps that identify the desired solution from the collected ball, which may contain both support and non-support vertices. Hence, to genuinely speed up the eigenwalk, a quantum algorithm would need to identify support vertices coherently from within this ball, not merely traverse edges faster, and there is no clear mechanism for this. This work thus establishes an insightful boundary between classical algorithms and current quantum algorithmic primitives. 

\newpage
Beyond the problem of computing exact sparse eigenvector supports, our algorithm also informs the parameter regimes in which existing quantum algorithms could provide an advantage for the weaker task of ground state approximation. For example, if a ground state is extremely sparse, then our algorithm may become practical despite its poor asymptotic dependence on the sparsity. Consequently, Hamiltonians with extremely sparse ground states are unlikely to be suitable candidates for demonstrations of quantum advantage using techniques such as sample-based Krylov quantum diagonalization \cite{skqd}, since they would be efficiently solvable classically. 

Lastly, there remain several open directions. Firstly, it is not obvious whether the non-extremal eigenwalk problem remains efficiently solvable without the sparsity-separation promise. It is also natural to ask whether more sophisticated eigenwalk algorithms could be devised that intelligently incorporate the numerical weights of the Hamiltonian to guide the search more efficiently. More broadly, however, this work suggests a useful perspective for algorithm design: eigenvectors can be studied not only as linear-algebraic objects, but also through the graph structure imposed on their supports. This perspective may inform classical or quantum algorithms beyond ground-state computation, wherever the relevant solution can be encoded by the support of a sparse eigenvector.

\newpage
\appendix
\section*{Appendix}
The appendices provide the deferred proofs and additional examples that support the results of the main text.
Appendix~\ref{appendix:proofs} collects the proofs of the key theorems: the localization theorem (Theorem~\ref{thm:localized}), the approximately sparse localization theorem (Theorem~\ref{thm:approx}), and the polynomial-time algorithm (Theorem~\ref{thm:alg}); it also includes a self-contained proof of the Rayleigh--Ritz theorem used throughout.
Appendix~\ref{appendix:examples} constructs explicit graph examples demonstrating the tightness of the bounds in Theorems~\ref{thm:degeneracy}, \ref{thm:eigstruct}, and~\ref{thm:localized}. Appendix~\ref{appendix:spca} explains how the localization result of Theorem~\ref{thm:localized} can be extended to the setting of sparse principal component analysis.

\section{Proofs of theorems from the main text}\label{appendix:proofs}
In this appendix, we provide the remaining proofs of theorems from the main text.
Section~\ref{appendix:general} proves Theorem~\ref{thm:localized} (localization of irreducible support components).
Section~\ref{appendix:approx} proves Theorem~\ref{thm:approx} (approximately sparse support localization).
Section~\ref{appendix:alg} proves Theorem~\ref{thm:alg} (the polynomial-time eigenwalk algorithm).
Section~\ref{appendix:rayleigh} provides a self-contained statement and proof of the Rayleigh--Ritz theorem, which is invoked in several of the preceding proofs. 

\subsection{Proof of Theorem \ref{thm:localized}}\label{appendix:general}
We now prove Theorem \ref{thm:localized}, restated below for convenience.
\setcounter{theorem}{2}
\begin{theorem}[Localization of irreducible support components]
    Let $G$ be a simple graph, and let $H$ be a $G$-consistent Hermitian matrix with an eigenvector $\ket{\psi}$ corresponding to eigenvalue $\lambda$. Set $s=\abs{\supp(\ket{\psi})}$. Then for every vertex $v\in\supp(\ket{\psi})$, there exists an eigenvector $\ket{\psi_v}$ of $H$ with eigenvalue $\lambda$ such that $v\in\supp(\ket{\psi_v})$ and
    \begin{align}
        \operatorname{diam}_G\bigl(\supp(\ket{\psi_v})\bigr) \leq
        \begin{cases}
            2s-2, & \text{if } \ket{\psi} \text{ is non-extremal},\\[4pt]
            s-1, & \text{if } \ket{\psi} \text{ is extremal},
        \end{cases}\nonumber
    \end{align}
    where $\operatorname{diam}_G(T):=\max_{i,j\in T}\operatorname{dist}(i,j)$ denotes the graph diameter of $T$ in $G$.
\end{theorem}
\begin{proof}
    Non-extremal case. Let $ R  = \supp(\ket{\psi})$, pick any $v\in  R $, and assume $\ket{\psi}$ is non-extremal. As a preliminary step, define a subgraph $G'$ of $G$ with vertices $N[ R ]$ and edges which have at least one endpoint in $ R $. More formally,
    \begin{align}\label{eq:G-prime}
        G' = (N[ R ],E')\ \mathrm{where}\ E' = \big\{\{j,j'\}\in E(G)\ :\ j\in  R , j'\in N[ R ]\big\}.
    \end{align}
    Throughout this proof, we will use $V(G)$ and $E(G)$ to denote the set of vertices and edges of $G$, respectively. The vertex $v$ lies in some connected component of $G'$; call this connected component $G'_1$, and call the union of the remaining connected components $G'_2$. Then define the vertex sets
     \begin{align}\label{eq:vertex-sets}
         S = V(G'_1)\cap R \ \mathrm{and}\ S' = V(G'_2)\cap  R .
     \end{align}
     Observe that $V(G'_1) = N[S]$ and $V(G'_2) = N[S']$ due to the construction of $G'$. First assume that $S'$ is nonempty. Evidently, $N[S]\cap N[S']=\emptyset$, otherwise there would be a vertex common to both $G'_1$ to $G'_2$, which is impossible since these two subgraphs are disconnected. It follows that $ R $ can be partitioned into nonempty $S$ and $S'$ and that $\operatorname{dist}(S,S')>2$; in fact, $S$ is the irreducible support component of $\ket{\psi}$ which contains $v$. Theorem \ref{thm:degeneracy} then implies that there exists an eigenvector $\ket{\psi_v}$ of $H$ with support $S$ and eigenvalue $\lambda$. Moreover, since $v\in S$, $v$ belongs to $\supp(\ket{\psi_v})$, as desired. If instead $S'$ were empty, then we would have $S = R$, and we could choose $\ket{\psi_v}=\ket{\psi}$, which has support equal to $S$ and contains $v$ in its support.

     It remains to be shown that the distance in $G$ between any $i,j\in S$ is $\leq 2s-2$. We first show that the distance in $G'_1$ between any $i,j\in S$ must be $\leq 2s-2$. Choose any $i,j\in S$. Since $G'_1$ is connected and $S\subseteq V(G_1')$, there is a shortest path $P$ of finite length connecting $i$ to $j$. We can write $P$ as $(k_1,\ldots,k_{\ell})$ where $\ell-1$ is the length (number of edges) of $P$, $k_1=i$, $k_{\ell}=j$, and $k_2,\ldots, k_{\ell-1}\in N[S]$. Note that for any $t\in\{1,\ldots,\ell-1\}$, if $k_t \in N(S)$, then $k_{t+1}$ must be in $S$ since the edge $\{k_{t},k_{t+1}\}$ must have at least one vertex in $S$ (by the definition of $G'_1$). 
     It follows that $P$ cannot contain more vertices from $N(S)$ than from $S$. Furthermore, since the endpoints of $P$ (i.e., $i$ and $j$) belong to $S$, if $P$ contains $q$ vertices from $S$, then it can contain at most $q-1$ vertices from $N(S)$. Hence,  $\ell \leq q+(q-1) = 2q-1$, and the length of $P$ is at most $(2q-1)-1 =   2q-2$. Furthermore, since $P$ is the shortest path connecting $i$ to $j$, there cannot be any repeated vertices in $P$. It follows that $q\leq |S|\leq s$, so the length of $P$ is at most $2s-2$. Lastly, because the distance in $G'_1$ between $i$ and $j$ is $\leq 2s-2$ and $G'_1$ is a subgraph of $G$, the distance between these two vertices in $G$ is also $\leq 2s-2$, as desired.

     Extremal case. Now assume $\ket{\psi}$ is extremal, and the proof proceeds quite analogously. Instead of defining $G'$ as in Eq. (\ref{eq:G-prime}), we define $G'$ to be the induced subgraph of $ R $ in $G$, that is, the graph with vertex set $ R $ and whose edge set contains all edges in $E(G)$ with both endpoints in $ R $. As before, $v$ lies in some connected component of $G'$ called $G'_1$, and the union of the remaining connected components is called $G'_2$. Again define the vertex sets $S$ and $S'$ according to Eq. (\ref{eq:vertex-sets}). This time, however, it is also true that $S= V(G'_1)$ and $S'= V(G'_2)$ since $V(G'_1),V(G'_2)\subseteq R $. First assume that $S'$ is nonempty. Evidently, $S\cap N(S') = \emptyset$ and $S'\cap N(S) = \emptyset$, otherwise there would be an edge in $G'$ connecting a vertex in $S$ to one in $S'$, which is impossible since the two subgraphs $G'_1$ and $G'_2$ are disconnected. It follows that $R$ can be partitioned into nonempty $S,S'$ and that $\operatorname{dist}(S,S')>1$; again, $S$ is the irreducible support component of $\ket{\psi}$ containing $v$. Theorem \ref{thm:degeneracy} then implies that there exists an eigenvector $\ket{\psi_v}$ of $H$ with support $S$ and eigenvalue $\lambda$. Additionally, $v$ is contained in $\supp(\ket{\psi_v})$, as desired. If instead $S'$ were empty, then we would have $S = R$, and we could choose $\ket{\psi_v}=\ket{\psi}$, which has support equal to $S$ and contains $v$ in its support.
     
     As above, it remains to be shown that the distance in $G$ between any $i,j\in S$ is $\leq s-1$. We first show that the distance in $G'_1$ between any $i,j\in S$ must be $\leq s-1$. Choose any $i,j\in S$. Since $G'_1$ is connected and $S= V(G_1')$, there is a shortest path $P$ connecting $i$ to $j$ which contains only non-repeated vertices in $S$. Thus, the length of $P$ is at most $|S|-1 \leq s-1$, since $|S|\leq s$. Furthermore, since $G'_1$ is a subgraph of $G$, the distance between $i$ and $j$ in $G$ is also $\leq s-1$, as desired.
\end{proof}

\subsection{Proof of Theorem \ref{thm:approx}}\label{appendix:approx}
Next, we prove Theorem \ref{thm:approx}, again restated below.
\begin{theorem}[Approximately sparse support localization]
    Let $G$ be a simple graph, and let $H$ be a $G$-consistent Hermitian matrix with a normalized extremal eigenvector $\ket{\psi}$ corresponding to eigenvalue $\lambda$. Suppose $\ket{\psi}$ is $(\eta,\epsilon)$-sparse over the vertex set $ A \subseteq V$, and let $s = \abs{ A }$. Then, given any vertex $v\in  A $, there exists an exactly sparse normalized state $\ket{\psi_v}$ with $v\in\supp(\ket{\psi_v})$ and $\supp(\ket{\psi_v})\subseteq  A $ such that
    \begin{align}
        \operatorname{diam}_G\bigl(\supp(\ket{\psi_v})\bigr) \leq s-1\nonumber
    \end{align}
    and 
    \begin{align}
        \abs{\bra{\psi_v}H\ket{\psi_v}-\lambda} \leq \frac{2\norm{H}}{\eta}\sqrt{\epsilon} + \mathcal{O}(\epsilon).\nonumber
    \end{align}
\end{theorem}

\begin{proof}
    Pick any $v\in A $. We define a set $S$ analogously to the $S$ constructed in the proof of the extremal case of Theorem \ref{thm:localized}: letting $G'$ be the induced subgraph of $ A $ in $G$, define $S$ to be the vertex set of the connected component of $G'$ that contains $v$. Parallel to this earlier proof, it is straightforward to show that the diameter of $S$ in $G$ is $\leq|S|-1$ and therefore $\leq s-1$, since $|S|\leq s$. Define $S'= A \setminus S$.
    Then $\ket{\psi}$ can be written as the linear combination of three normalized vectors $\ket{\psi_S}$, $\ket{\psi_{S'}}$, and $\ket{\psi^\perp}$ supported on $S$, $S'$, and $V\setminus  A $, respectively:
    \begin{align}
        \ket{\psi} = \alpha\ket{\psi_S}+\beta\ket{\psi_{S'}} + \gamma\ket{\psi^\perp}, 
    \end{align}
    where $\alpha$, $\beta$, and $\gamma$ are some scalars whose square magnitudes sum to 1 (without loss of generality, choose $\alpha$ to be nonnegative). Since the diameter of $S$ in $G$ is $\leq s-1$, the choice $\ket{\psi_v}=\ket{\psi_S}$ satisfies Eq. (\ref{eq:approx-cluster}). We will now prove that $\ket{\psi_v}=\ket{\psi_S}$ also satisfies Eq. (\ref{eq:approx-energy}) using the following strategy:
    \begin{enumerate}
        \item Define the normalized state 
        \begin{align}
            \ket{\phi} = \frac{1}{\sqrt{|\alpha|^2 + |\beta|^2}} (\alpha\ket{\psi_S}+\beta\ket{\psi_{S'}}). 
        \end{align}
        Since $\ket{\phi}$ and the true eigenvector $\ket{\psi}$ have large overlap, the expectation values of their energies are close. 
        \item Now define the normalized state
        \begin{align}
            \ket{\phi'} = \frac{1}{\sqrt{|\alpha|^2 + |\beta|^2}} (\alpha\ket{\psi_S}-\beta\ket{\psi_{S'}}).
        \end{align}
        By the definition of $S$, the subsets $S$ and $S'$ must belong to different connected components of the induced subgraph of $G$ on $ A $. Hence, on the graph $G$, none of the vertices in $S$ are neighbors of those in $S'$, and vice versa. This fact is sufficient to guarantee that $\ket{\phi'}$ and $\ket{\phi}$ have identical expectation values of energy.
        \item Observe that $\ket{\psi_S}$ is an equal superposition of $\ket{\phi}$ and $\ket{\phi'}$. Because of the condition that $\abs{\psi_j}^2\geq \eta$ for all $j\in  A $, the magnitude of $\alpha$ (i.e., the coefficient of $\ket{\psi_S}$ projected onto each of $\ket{\phi}$ and $\ket{\phi'}$) cannot be too small. Together with the results of steps 1 and 2, this fact implies that the expectation value of energy $\ket{\psi_S}$ is close to $\lambda$, thereby implying Eq. (\ref{eq:approx-energy}).
    \end{enumerate}
    Figure~\ref{fig:thm4proof} illustrates the graph decomposition (panel a) and state-space geometry (panel b) underlying these three steps.
    We now fill in the details of each of these steps sequentially, beginning with step 1.  
    Note that
    \begin{align}
        \abs{\bra{\phi}H\ket{\phi}-\lambda} &=  \abs{\bra{\phi}(H-\lambda)\ket{\phi}}.
    \end{align}
    Application of the Cauchy-Schwarz inequality above gives
    \begin{align}
        \abs{\bra{\phi}H\ket{\phi}-\lambda} &\leq \norm{\ket{\phi}}\norm{(H-\lambda)\ket{\phi}}\\
        &= \norm{(H-\lambda)\ket{\phi}}
    \end{align}
    where $\norm{\cdot}$ represents the 2-norm for vector arguments. Observe that since $\ket{\psi}$ is an eigenvector with eigenvalue $\lambda$, $(H-\lambda)\ket{\psi} = 0$. It thus follows that 
    \begin{align}
         \abs{\bra{\phi}H\ket{\phi}-\lambda}&\leq \norm{(H-\lambda)(\ket{\phi}-\ket{\psi})}\\
         &\leq \norm{H-\lambda}\norm{\ket{\phi}-\ket{\psi}}\\
         &\leq 2\norm{H}\norm{\ket{\phi}-\ket{\psi}},
    \end{align}
    where $\norm{\cdot}$ represents the operator norm for matrix arguments. Hence if $\ket{\psi}$ and $\ket{\phi}$ are close, their expectation energies must be close too. We now explicitly verify that the vectors are indeed close:
    \begin{align}
        \norm{\ket{\psi}-\ket{\phi}}^2 &= \left(|\alpha|^2+|\beta|^2\right)\left(1-\frac{1}{\sqrt{|\alpha|^2+|\beta|^2}}\right)^2 + |\gamma|^2\\
        &= \left(|\alpha|^2+|\beta|^2\right)\left(1-\frac{2}{\sqrt{|\alpha|^2+|\beta|^2}}+\frac{1}{|\alpha|^2+|\beta|^2}\right) + |\gamma|^2.
    \end{align}
    Since $|\alpha|^2+|\beta|^2 = 1-|\gamma|^2$, the above expression simplifies as
    \begin{align}
        \norm{\ket{\psi}-\ket{\phi}}^2 &= 2-2\sqrt{1-|\gamma|^2}.
    \end{align}
    Hence, we have completed step 1, obtaining the result
    \begin{align}
        \abs{\bra{\phi}H\ket{\phi}-\lambda} &\leq 2\norm{H} \sqrt{2-2\sqrt{1-|\gamma|^2}}\\
        &= 2\norm{H}|\gamma| + \mathcal{O}(|\gamma|^2).
    \end{align}

    Next, we proceed to step 2. Note that the application of $H$ to $\ket{\psi_S}$ yields a vector whose support is contained in $N[S] = S \cup N(S)$. Let $\Pi_S = \sum_{j\in S}\ketbra{j}$ and $\Pi_{N(S)}=\sum_{j\in N(S)}\ketbra{j}$ be the projectors onto the spans of $S$ and $N(S)$, respectively. Then 
    \begin{align}
        H\ket{\psi_S} = \left(\Pi_S+\Pi_{N(S)}\right) H \ket{\psi_S}.
    \end{align}
    Analogously, 
    \begin{align}
        H\ket{\psi_{S'}} = \left(\Pi_{S'}+\Pi_{N(S')}\right) H \ket{\psi_{S'}}.
    \end{align}
    The energy expectation for $\ket{\phi}$ can then be computed accordingly:
    \begin{align}
        \bra{\phi}H\ket{\phi} &= \frac{1}{|\alpha|^2+|\beta|^2} \left(\alpha^*\bra{\psi_S}+\beta^*\bra{\psi_{S'}}\right)\bigg(\alpha \left(\Pi_S+\Pi_{N(S)}\right) H \ket{\psi_S} \nonumber\\
        &\qquad + \beta \left(\Pi_{S'}+\Pi_{N(S')}\right) H \ket{\psi_{S'}}\bigg)\\
        &= \frac{1}{|\alpha|^2+|\beta|^2} \left(|\alpha|^2\bra{\psi_S}H\ket{\psi_S} + |\beta|^2\bra{\psi_{S'}}H\ket{\psi_{S'}}\right),\label{eq:e-expectation}
    \end{align}
    where the second equality follows since 
    \begin{align}
        0&=\bra{\psi_S}\Pi_{N(S)} = \bra{\psi_{S'}}\Pi_{N(S)} = \bra{\psi_S}\Pi_{N(S')} = \bra{\psi_{S'}}\Pi_{N(S')}=\bra{\psi_{S'}}\Pi_S = \bra{\psi_S}\Pi_{S'}
    \end{align}
    and $\bra{\psi_S}\Pi_S = \bra{\psi_S}$ and $\bra{\psi_{S'}}\Pi_{S'}=\bra{\psi_{S'}}$. Analogously computing the energy expectation for $\ket{\phi'}$ reveals that $\bra{\phi'}H\ket{\phi'}$ also equals the quantity in Eq. (\ref{eq:e-expectation}). Hence, $\bra{\phi}H\ket{\phi} = \bra{\phi'}H\ket{\phi'}$.

    Lastly, we address step 3. Namely, we show that $\bra{\psi_S}H\ket{\psi_S}$ is close to $\lambda$. Note that $\ket{\psi_S}$ can be written as the following linear combination of $\ket{\phi}$ and $\ket{\phi'}$:
    \begin{align}
        \ket{\psi_S} = \mathcal{N}(\ket{\phi}+\ket{\phi'}),\ \mathrm{where}\ \mathcal{N} = \frac{\sqrt{|\alpha|^2+|\beta|^2}}{2|\alpha|}.
    \end{align}
    Then the difference between $\lambda$ and the expectation energy of $\ket{\psi_S}$ is given by
    \begin{align}
        \abs{\bra{\psi_S}H\ket{\psi_S}- \lambda} &= \abs{\bra{\psi_S}(H-\lambda)\ket{\psi_S}} \\
        &= \mathcal{N}^2 \Big|\big(\bra{\phi}(H-\lambda)\ket{\phi} + \bra{\phi'}(H-\lambda)\ket{\phi'} + \bra{\phi}(H-\lambda)\ket{\phi'}\\
        &\qquad+\bra{\phi'}(H-\lambda)\ket{\phi}\big)\Big|,
    \end{align}
    and it follows from the triangle inequality that
    \begin{align} \label{eq:triangle}
        \abs{\bra{\psi_S}H\ket{\psi_S}-\lambda} \leq \mathcal{N}^2 \big(|\bra{\phi}(H-\lambda)\ket{\phi}| + |\bra{\phi'}(H-\lambda)\ket{\phi'}| + 2|\bra{\phi'}(H-\lambda)\ket{\phi}|\big).
    \end{align}
    Steps 1 and 2 have already yielded upper bounds for the first two terms in this sum. We must now upper bound the remaining cross term. 
    Applying the Cauchy-Schwarz inequality to the cross term, we have
    \begin{align}
        \abs{\bra{\phi'}(H-\lambda)\ket{\phi}} &\leq \norm{\ket{\phi'}}\norm{(H-\lambda)\ket{\phi}}\\
        &= \norm{(H-\lambda)\ket{\phi}}.
    \end{align}
    Recall that in step 1, we already showed that $\norm{(H-\lambda)\ket{\phi}} \leq 2\norm{H}\sqrt{2-2\sqrt{1-|\gamma|^2}}$. Hence 
    \begin{align}
        \abs{\bra{\phi'}(H-\lambda)\ket{\phi}} &\leq  2\norm{H}\sqrt{2-2\sqrt{1-|\gamma|^2}}.
    \end{align}
    We now substitute this cross term into Eq. (\ref{eq:triangle}):
    \begin{align}
        \abs{\bra{\psi_S}H\ket{\psi_S}-\lambda} &\leq 8\mathcal{N}^2 \norm{H}\sqrt{2-2\sqrt{1-|\gamma|^2}}\\
        &= 2\frac{|\alpha|^2+|\beta|^2}{|\alpha|^2} \norm{H}\sqrt{2-2\sqrt{1-|\gamma|^2}}\\
        &= \frac{2 \norm{H}}{|\alpha|^2} (1-|\gamma|^2)\sqrt{2-2\sqrt{1-|\gamma|^2}},\label{eq:upper-bound-gamma}
    \end{align}
    where in the last equality we have used the fact that $|\alpha|^2+|\beta|^2 = 1-|\gamma|^2$. Now observe that $|\gamma|^2 \leq \epsilon$ and $|\alpha|^2 \geq \eta$. The quantity in Eq. (\ref{eq:upper-bound-gamma}) is monotonically increasing in $\abs{\gamma}^2$ over the interval $[0,\frac{9}{25})$, reaching a maximum at the latter endpoint. Since $\epsilon$ lies in this interval by assumption, we conclude that 
    \begin{align}
        \abs{\bra{\psi_S}H\ket{\psi_S}-\lambda} &\leq \frac{2 \norm{H}}{\eta} (1-\epsilon)\sqrt{2-2\sqrt{1-\epsilon}}.
    \end{align}
    Eq. (\ref{eq:approx-energy}) is recovered by choosing $\ket{\psi_v} = \ket{\psi_S}$ and keeping the leading-order term in $\epsilon$.
\end{proof}


\begin{figure}[htbp]
\centering
\resizebox{\columnwidth}{!}{%
\begin{tikzpicture}[
  every node/.style={circle, draw=black, fill=white, inner sep=0pt,
                     minimum size=17pt, line width=0.6pt},
  vnode/.style  = {fill=black!80},
  supp/.style   = {fill=black!75},
  suppB/.style  = {fill=black!45},
  leaked/.style = {draw=black!35, fill=white},
  lbl/.style    = {draw=none, fill=none, rectangle,
                   font=\small, inner sep=2pt},
  lbl2/.style   = {draw=none, fill=none, rectangle,
                   font=\footnotesize, inner sep=2pt},
  lbl3/.style   = {draw=none, fill=none, rectangle,
                   font=\scriptsize, inner sep=1.5pt},
  sboxA/.style  = {draw=blue!55!black,   fill=blue!9,    rectangle,
                   rounded corners=2.5pt, font=\footnotesize,
                   inner sep=4pt, align=left, line width=0.55pt},
  sboxB/.style  = {draw=green!52!black,  fill=green!10,  rectangle,
                   rounded corners=2.5pt, font=\footnotesize,
                   inner sep=4pt, align=left, line width=0.55pt},
  sboxC/.style  = {draw=orange!65!black, fill=orange!12, rectangle,
                   rounded corners=2.5pt, font=\footnotesize,
                   inner sep=4pt, align=left, line width=0.55pt},
  ldrA/.style = {draw=blue!50!black,   line width=0.55pt, dashed},
  ldrB/.style = {draw=green!50!black,  line width=0.55pt, dashed},
  ldrC/.style = {draw=orange!60!black, line width=0.55pt, dashed},
  aedge/.style = {draw=black!35, line width=0.5pt},
]

\begin{scope}

\draw[draw=black!55, dashed, line width=0.8pt, fill=black!5, rounded corners=28pt]
  (-1.5, 0.45) rectangle (4.9, 4.4);

\draw[draw=black!65, line width=0.7pt, fill=black!11, rounded corners=16pt]
  (-1.2, 0.95) rectangle (1.8, 3.95);
\draw[draw=black!65, line width=0.7pt, fill=black!11, rounded corners=16pt]
  (2.4, 1.10) rectangle (4.6, 3.80);

\node[leaked] (lk1) at (-2.4, 2.8) {};
\node[leaked] (lk2) at (-2.4, 1.7) {};
\node[leaked] (lk3) at ( 5.7, 2.5) {};

\draw[aedge] (-0.6, 2.7) -- (lk1);
\draw[aedge] (-0.6, 1.6) -- (lk2);
\draw[aedge] (4.3,  3.0) -- (lk3);   

\node[supp]  (s1) at (-0.3, 2.7) {};
\node[vnode] (v)  at ( 0.7, 2.2) {};
\node[supp]  (s2) at (-0.3, 1.6) {};
\node[supp]  (s3) at ( 0.9, 3.2) {};
\draw (s1)--(v); \draw (s2)--(v); \draw (s1)--(s2); \draw (s3)--(v); \draw (s3)--(s1);

\node[suppB] (t1) at (2.9, 3.0) {};
\node[suppB] (t2) at (4.1, 3.0) {};
\node[suppB] (t3) at (3.5, 1.8) {};
\draw (t1)--(t2); \draw (t2)--(t3); \draw (t1)--(t3);

\node[lbl2] at ( 0.3, 3.62) {$S$};
\node[lbl2] at ( 3.55, 3.62) {$S'$};
\node[lbl2, text=black!55] at (-2.42, 3.40) {$V\!\setminus\!A$};
\node[lbl3,text=white] at ( 0.7, 2.2) {$v$};

\node[draw=none, fill=none, rectangle, font=\scriptsize, text=black!38,
      inner sep=1pt] at (1.70, 0.66) {$A$ (approx.\ support)};

\node[lbl2, anchor=north] at (1.70, 0.25)
  {$\ket{\psi}=\alpha\ket{\psi_S}+\beta\ket{\psi_{S'}}+\gamma\ket{\psi^\perp},
    \quad |\gamma|^2\leq\epsilon$};

\node[lbl, anchor=south] at (1.70, 4.68)
  {\textbf{(a)}\;graph decomposition of approximate support $A$};

\end{scope}

\begin{scope}[xshift=7.0cm, yshift=1.86cm]

\draw[draw=black!18, line width=0.4pt] (-0.40, 0) -- (3.05, 0);
\draw[draw=black!18, line width=0.4pt] (0, -2.65) -- (0, 2.70);
\node[lbl2, anchor=south west, text=black!55] at (0.01, 2.20)
  {$\ket{\psi_{S'}}$};


\draw[->, line width=1.3pt, draw=blue!65!black] (0,0) -- (1.84, 1.54);
\node[lbl2, anchor=north west, text=blue!72!black] at (1.60, 1.35)
  {$\ket{\phi}$};

\draw[->, line width=1.3pt, draw=orange!75!black] (0,0) -- (1.84, -1.54);
\node[lbl2, anchor=north, text=orange!82!black] at (1.20, -1.20)
  {$\ket{\phi'}$};

\draw[->, line width=1.8pt, draw=green!52!black] (0,0) -- (3.00, 0);
\node[lbl2, anchor=north, text=green!52!black] at (3.00, -0.12)
  {$\ket{\psi_S}$};

\draw[->, line width=0.9pt, draw=black!42, dashed] (0,0) -- (1.00, 1.73);
\node[lbl3, anchor=south west, text=black!50] at (1.02, 1.76)
  {$\ket{\psi}$\;$(|\gamma|^2\!\leq\!\epsilon)$};

\draw[draw=black!42, line width=0.5pt] (0.65,0) arc(0: 40:0.65cm);
\node[lbl3, text=black!58] at (0.88, 0.22) {$\theta$};
\draw[draw=black!42, line width=0.5pt] (0.65,0) arc(0:-40:0.65cm);
\node[lbl3, text=black!58] at (0.88,-0.22) {$\theta$};


\node[sboxA, anchor=west] (box1) at (3.90, 1.54)
  {\tikz[baseline=(n.base),every node/.style={}]
     \node[circle,draw=blue!55!black,inner sep=1.1pt,
           font=\footnotesize\bfseries,line width=0.55pt](n){1};%
   ~Cauchy-Schwarz:\\
   $\ket{\phi}\approx\ket{\psi}$\;$(|\gamma|^2\leq\epsilon)$\\
   $\Rightarrow|\bra{\phi}H\ket{\phi}-\lambda|\leq 2\|H\||\gamma|$};

\node[sboxB, anchor=west] (box3) at (3.90, 0.00)
  {\tikz[baseline=(n.base),every node/.style={}]
     \node[circle,draw=green!52!black,inner sep=1.1pt,
           font=\footnotesize\bfseries,line width=0.55pt](n){3};%
   ~$\ket{\psi_S}=\mathcal{N}(\ket{\phi}+\ket{\phi'})$\\
   $\Rightarrow|\bra{\psi_S}H\ket{\psi_S}-\lambda|\leq
     \tfrac{2\|H\|}{\eta}\sqrt{\epsilon}+\mathcal{O}(\epsilon)$};

\node[sboxC, anchor=west] (box2) at (3.90, -1.54)
  {\tikz[baseline=(n.base),every node/.style={}]
     \node[circle,draw=orange!65!black,inner sep=1.1pt,
           font=\footnotesize\bfseries,line width=0.55pt](n){2};%
   ~$S\!\perp\!S'$ in $G[A]$ (no cross edges):\\
   $\Rightarrow\bra{\phi}H\ket{\phi}=\bra{\phi'}H\ket{\phi'}$};

\draw[ldrA] (box1.west) -- (1.86, 1.54);   
\draw[ldrB] (box3.west) -- (3.04, 0.00);   
\draw[ldrC] (box2.west) -- (1.86, -1.54);  

\node[lbl, anchor=south west] at (0.20, 2.82)
  {\textbf{(b)}\;state-space decomposition};

\end{scope}

\end{tikzpicture}%
}

\caption{Proof sketch of Theorem~\ref{thm:approx} (approximately sparse support localization).
\textbf{(a)}~The extremal eigenvector $\ket{\psi}$ with eigenvalue $\lambda$ is approximately supported on the vertex subset $A$. Within the induced subgraph of $A$ in $G$, the guiding vertex $v$ belongs to a connected component whose vertex set is $S$; the remaining vertices in $A$ belong to $S'$. Vertices outside $A$ carry the leakage amplitude $\gamma$, with $|\gamma|^2\leq\epsilon$. The localized states $\ket{\psi_S}$ and $\ket{\psi_{S'}}$ are the projections of $\ket{\psi}$ to $S$ and $S'$, respectively. 
\textbf{(b)}~State-space picture (vertical axis: $\ket{\psi_{S'}}$; horizontal: $\ket{\psi_S}$).
\textcircled{\raisebox{-0.8pt}{\footnotesize 1}}~(blue) $\ket{\phi}$ is the normalized projection
of $\ket{\psi}$ to the approximate support $A$; Cauchy-Schwarz bounds how far its energy deviates from $\lambda$.
\textcircled{\raisebox{-0.8pt}{\footnotesize 2}}~(orange) Reflecting $\ket{\phi}$ across $\ket{\psi_S}$ gives $\ket{\phi'}$;
disconnection of $S$ and $S'$ in the induced subgraph forces $\ket{\phi'}$ to have the same energy
expectation as $\ket{\phi}$.
\textcircled{\raisebox{-0.8pt}{\footnotesize 3}}~(green) The localized state $\ket{\psi_S}\propto\ket{\phi}+\ket{\phi'}$
lies along the bisector; combining
\textcircled{\raisebox{-0.8pt}{\footnotesize 1}} and~\textcircled{\raisebox{-0.8pt}{\footnotesize 2}} bounds its energy
deviation from $\lambda$ by $(2\|H\|/\eta)\sqrt{\epsilon}+\mathcal{O}(\epsilon)$.}
\label{fig:thm4proof}
\end{figure}
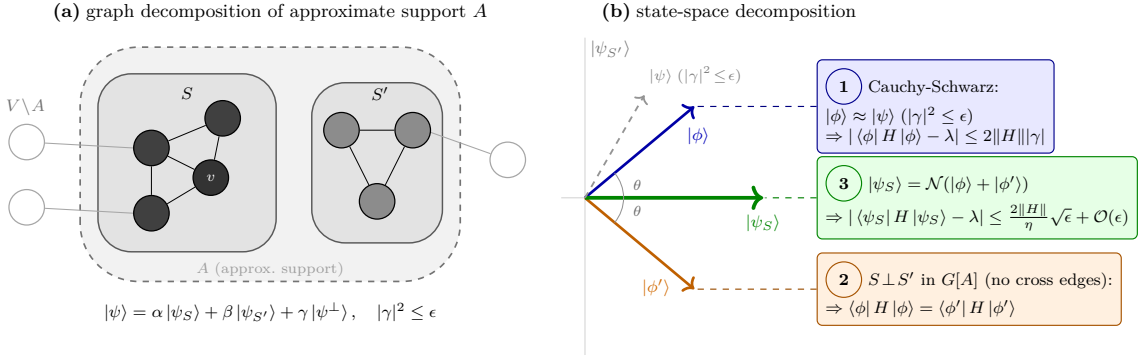

Recall that for the bound in Eq. (\ref{eq:approx-energy}) of Theorem~\ref{thm:approx} to not diverge as $\eta,\epsilon\rightarrow 0$, it is necessary for $\eta$ to scale as $\Omega(\sqrt{\epsilon})$. This assumption is not always physically realistic, but an alternative mild assumption allows for the $\eta$ dependence to be removed entirely in some cases. Namely, 
suppose that $\abs{\alpha}> \abs{\beta}$; in other words, after partitioning the approximate support $ A $ into subsets $S$ and $S'$, more of the weight of $\ket{\psi}$ lies in $S$---the part which contains $v$---than in the complement $S'$. This assumption is reasonable when $v$ is given as a guiding vertex that is promised to have large weight in $\ket{\psi}$.

We now demonstrate how this assumption removes the dependence on $\eta$. Note that $\abs{\alpha}^2+\abs{\beta}^2 = 1-\abs{\gamma}^2 \geq 1-\epsilon$. Hence, $\abs{\beta}^2 \geq 1-\epsilon - \abs{\alpha}^2$. Since $\abs{\alpha}^2>\abs{\beta}^2$ by assumption, we have $\abs{\alpha}^2 > 1-\epsilon -\abs{\alpha}^2$. It follows that
\begin{align}\label{eq:alph-lower-bound}
    \abs{\alpha}^2 > (1-\epsilon)/2.
\end{align}
Substituting this inequality and $\abs{\gamma}^2\leq \epsilon$ into Eq. (\ref{eq:upper-bound-gamma}), we obtain
\begin{align}
    \abs{\bra{\psi_S}H\ket{\psi_S}-\lambda} &\leq 4 \norm{H}\sqrt{2-2\sqrt{1-\epsilon}}.
\end{align}
Keeping the leading-order term in $\epsilon$, we recover
\begin{align}
    \abs{\bra{\psi_S}H\ket{\psi_S}-\lambda} \leq 4\norm{H}\sqrt{\epsilon} + \mathcal{O}(\epsilon).
\end{align}
Therefore, if the assumption $\abs{\alpha}>\abs{\beta}$ is added to Theorem~\ref{thm:approx}, Eq.~(\ref{eq:approx-energy}) can be replaced with 
\begin{align}
    \abs{\bra{\psi_v}H\ket{\psi_v}-\lambda} \leq 4\norm{H}\sqrt{\epsilon} + \mathcal{O}(\epsilon),
\end{align}
which has no explicit $\eta$ dependence.

\subsection{Proof of Theorem \ref{thm:alg}}\label{appendix:alg}
We prove Theorem~\ref{thm:alg} below.
\begin{proof}[Proof of Theorem~\ref{thm:alg}]
    Algorithm~\ref{algorithm:eigenwalk} below solves the sparse guided eigenwalk problem. Recall that due to Theorem~\ref{thm:localized}, there must exist a solution eigenvector whose support is localized near the guiding vertex. The algorithm collects all vertices within a sufficiently large radius of the guiding vertex such that the support of this eigenvector is guaranteed to be contained within the collected vertices. Then, the Hamiltonian is projected to the span of these collected vertices and exactly diagonalized. In the extremal case, the solution can be straightforwardly identified as an extremal eigenvector of the projected Hamiltonian (due to the Rayleigh--Ritz theorem). 

    However, identifying the solution is more subtle in the non-extremal case. First, it is necessary to determine which eigenvectors of the projected Hamiltonian are also eigenvectors of the full Hamiltonian. To do so, the energy variance (with respect to the full Hamiltonian) is computed for every eigenvector of the projected Hamiltonian; those vectors with zero variance must then also be eigenvectors of the full Hamiltonian. Group these eigenvectors of the full Hamiltonian into eigenspaces according to their eigenvalues. Then, the solution can be identified by leveraging the sparsity-separation promise of the non-extremal setting: since this promise guarantees that any $s$-sparse eigenvector of the full Hamiltonian must be a solution, it is sufficient to check if any of these eigenspaces contains an $s$-sparse vector. 
    \begin{algorithm}[H]
    \caption{Sparse guided eigenwalk}\label{algorithm:eigenwalk}
        \begin{algorithmic}[1]
            \Require Simple graph $G=(V,E)$ of maximum degree $d$, a positive integer $2\leq s\leq |V|$ (eigenvector sparsity), positive integers $m\leq |V|$ (total number of eigenvalues) and $k\leq m$ (eigenvalue index), sparse row access to a $G$-consistent Hermitian matrix $H$ with distinct eigenvalues $\lambda_1,\ldots,\lambda_m$ promised to have an eigenvector $\ket{\psi}$ with eigenvalue $\lambda_k$ and $\abs{\supp(\ket{\psi})}\leq s$, and a guiding vertex $v$ promised to belong to $\supp(\ket{\psi})$. If $k\notin\{1,m\}$, assume that the support of every eigenvector with eigenvalue $\neq\lambda_k$ has size $>s$.
            \Ensure The support of an eigenvector $\ket{\psi'}$ of $H$ with eigenvalue $\lambda_k$ that satisfies $\abs{\supp{(\ket{\psi'})}} \leq s$, in time $\mathcal{O}(d^{3s-3})$ (extremal case) or $\mathcal{O}(d^{(2s-2)(s+3)})$ (non-extremal case).
            \State Define $r = s-1$ if $k\in\{1,m\}$ and $r=2s-2$ if $k\notin\{1,m\}$.
            \State Compute the ball $B=\{u\in V:\operatorname{dist}(u,v)\leq r\}$ via breadth-first search starting at $v$ and truncating at depth $r$. 
            \State Determine $\Pi_BH\Pi_B$, where $\Pi_B=\sum_{j\in B}\ketbra{j}$ is the projector onto the span of $B$.
            \State Perform exact diagonalization on $\Pi_BH\Pi_B$ restricted to $\operatorname{Span}\{\ket{j}:j\in B\}$, obtaining all of its eigenvalues and eigenvectors.
            \If{$k\in\{1,m\}$,}
                \State If $k=1$ ($k=m$), choose any lowest (highest) eigenvector $\ket{\psi'}$ of $\Pi_BH\Pi_B$, and return $\supp(\ket{\psi'})$.
            \Else 
                \For{each eigenvalue $\mu$ of $\Pi_BH\Pi_B$,}
                    \State Initialize $R_\mu\leftarrow \{\}$, set of eigenvectors of the full $H$ with eigenvalue $\mu$.
                    \For{each normalized eigenvector $\ket{\phi}$ of $\Pi_BH\Pi_B$ with eigenvalue $\mu$}
                        \State Compute $\operatorname{Var}(\ket{\phi}) = \bra{\phi}H^2\ket{\phi}-\bra{\phi}H\ket{\phi}^2$.
                        \If{$\operatorname{Var}(\ket{\phi}) = 0$}
                            \State $R_\mu\leftarrow R_\mu \cup \{\ket{\phi}\}$
                        \EndIf
                    \EndFor
                    \State Using Subroutine~\ref{subroutine:sparse}, determine whether the span of $R_\mu$ contains an $s$-sparse vector, and return the support of this vector if True.  
                \EndFor
            \EndIf 
        \end{algorithmic}
    \end{algorithm}

    Algorithm \ref{algorithm:eigenwalk} utilizes the subroutine below to determine whether an $s$-sparse vector is contained within the span of a given set of vectors. Let $V$ be the matrix whose columns are the given vectors. Suppose such an $s$-sparse vector $\ket{\phi}$ exists. Since $\ket{\phi}$ is contained in the column span of $V$, it can be written as $\ket{\phi} = Vw$ for some column vector $w$. Furthermore, because $\ket{\phi}$ is $s$-sparse, there is a size-$s$ subset $S$ of indices which contains $\supp(\ket{\phi})$. Form a submatrix $V_{S^c}$ of $V$ by deleting all rows corresponding to indices in $S$. Evidently, $V_{S^c}w = 0$, but $Vw \neq 0$. Since any vector in the kernel of $V$ must also be in the kernel of $V_{S^c}$, it follows that $\ker V_{S^c}\supset \ker V$; hence $\rank V_{S^c} < \rank V$ by the rank-nullity theorem. 

    Thus, if a satisfactory $s$-sparse $\ket{\phi}$ exists, then there exists some subset $S$ of indices such that $\rank V_{S^c} < \rank V$. The converse statement also holds: if the rank condition is satisfied for some $S$, then there is some $w$ in $\ker V_{S^c}$ but not in $\ker V$, and $\ket{\phi}=Vw$ must be $s$-sparse. Consequently, the rank condition is necessary and sufficient for the existence of the desired sparse state. The subroutine then simply checks whether the rank condition holds for any of the possible size-$s$ sets of indices.
    
    \begin{subroutine}[H]
    \caption{Determine whether $s$-sparse vector is contained in span of given vectors}\label{subroutine:sparse}
        \begin{algorithmic}[1]
            \Require Set of vectors $\ket{\phi_1},\ldots,\ket{\phi_M}\in\mathbb{C}^N$ and positive integer $s$ (desired sparsity).
            \Ensure Decide if there exists any vector $\ket{\phi}$ in the span of $\ket{\phi_1},\ldots,\ket{\phi_M}\in\mathbb{C}^N$ whose support in the computational basis has size $\leq s$. If true, return any such $\ket{\phi}$. Runs in time $\mathcal{O}(N^{s+1}M\min\{N,M\})$.
            \State Form the matrix $V\in\mathbb{C}^{N\times M}$ whose columns are $\ket{\phi_1},\ldots \ket{\phi_M}$.
            \State Compute $r=\rank V$.
            \For{each subset $S\subseteq\{1,\ldots,N\}$ with $\abs{S} = s$}
                \State Form the submatrix $V_{S^c}$ of $V$ containing only the rows outside $S$.
                \If{$\rank{V_{S^c}}< r$}
                    \State Find $w$ in $\ker V_{S^c}$ that is not in $\ker V$.
                    \State Return True with solution $\ket{\phi}= Vw$.
                \EndIf
            \EndFor
            \State Otherwise, return False.
        \end{algorithmic}
    \end{subroutine}

    We now prove the correctness of Algorithm~\ref{algorithm:eigenwalk}, beginning with the extremal case. It suffices to show that any lowest eigenvector of $\Pi_BH\Pi_B$ restricted to $\operatorname{Span}\{\ket{j}:j\in B\}$ is also a lowest eigenvector of $H$ (the case of a highest eigenvector follows identically). First note that Theorem~\ref{thm:localized} implies that there exists a normalized lowest eigenvector $\ket{\psi_v}$ of $H$ whose support is contained within the ball $B$ of radius $s-1$ centered at $v$. Since $\supp(\ket{\psi_v})\subseteq B$, we have $\ket{\psi_v}\in\operatorname{Span}\{\ket{j}:j\in B\}$ and $\Pi_B\ket{\psi_v}=\ket{\psi_v}$. Therefore,
    \begin{align}
        \bra{\psi_v}\Pi_B H\Pi_B\ket{\psi_v}=\bra{\psi_v}H\ket{\psi_v}.
    \end{align}
    Moreover, for every normalized $\ket{\phi}\in\operatorname{Span}\{\ket{j}:j\in B\}$, we have $\Pi_B\ket{\phi}=\ket{\phi}$ and 
    \begin{align}
        \bra{\phi}\Pi_BH\Pi_B\ket{\phi} = \bra{\phi}H\ket{\phi}\geq \lambda,
    \end{align}
    where $\lambda$ is the lowest eigenvalue of $H$. Since equality is achieved by $\ket{\psi_v}$, the lowest eigenvalue of $\Pi_BH\Pi_B$ restricted to $\operatorname{Span}\{\ket{j}:j\in B\}$ is also $\lambda$. Therefore, for any lowest eigenvector $\ket{\psi'}$ of this restricted operator, $\bra{\psi'}H\ket{\psi'} = \lambda$. By the Rayleigh--Ritz theorem (see Appendix \ref{appendix:rayleigh}), any such $\ket{\psi'}$ is a lowest eigenvector of $H$ as well, as desired.

    In the non-extremal case, it is sufficient to show that an eigenvector $\ket{\phi}$ of $\Pi_B H \Pi_B$ restricted to $\operatorname{Span}\{\ket{j}:j\in B\}$ is an eigenvector of $H$ if and only if the energy variance of $\ket{\phi}$ with respect to $H$ is zero. The forward direction is clear. To prove the reverse direction, assume that a normalized eigenvector $\ket{\phi}$ of the restricted $\Pi_B H \Pi_B$ satisfies $0=\operatorname{Var}(\ket{\phi}) = \bra{\phi}H^2\ket{\phi}-\bra{\phi}H\ket{\phi}^2$. Let $\mu =\bra{\phi}H\ket{\phi}$. Then, since $H=H^\dagger$, we have ${0 = \bra{\phi}(H-\mu)^2\ket{\phi} = \norm{(H-\mu)\ket{\phi}}}$. It follows that $(H-\mu)\ket{\phi} = 0$, so $H\ket{\phi} = \mu\ket{\phi}$ and $\ket{\phi}$ is an eigenvector of $H$, as desired.

    We now prove that Algorithm~\ref{algorithm:eigenwalk} runs in the stated time. We begin analyzing the runtime of the breadth-first search. The search runs in time $\mathcal{O}\left(\abs{B}+\abs{E(B)}\right)$, where $E(B)$ is the set of edges with at least one vertex in $B$. If $d<2$, then $\abs{B} = \mathcal{O}(1)$; otherwise, 
    \begin{align}
        \abs{B} \leq 1+d + d(d-1) +\ldots+ d(d-1)^{r-1} = \mathcal{O}(d^{r}),
    \end{align}
    since a ball of radius $r$ centered at $v$ has 1 vertex at distance 0 from the center, at most $d$ vertices at distance 1, and at most $d-1$ vertices at distances $2,\ldots,r$. Furthermore, note that $E(B)\leq d\abs{B}/2$, since there are at most $d$ edges per vertex, with a factor of $1/2$ to prevent double-counting. Thus, the breadth-first search runs in time $\mathcal{O}(d^{r + 1})$.

    Next we determine the runtime of computing the projection $\Pi_B H \Pi_B$. Note that since $H$ is $G$-consistent and $G$ has maximum degree $d$, $H$ has at most $d+1$ nonzero entries per row and column. We then compute the entries of the projection as follows. For every $j\in B$, query the at most $d+1$ nonzero entries in row $j$ of $H$, and keep only those entries whose column index also lies in $B$. With worst-case $\mathcal{O}(|B|)$ lookup for membership in $B$, computing the entries of the projection takes $\mathcal{O}(d\abs{B}^2) = \mathcal{O}(d^{2r+1})$ time. 

    Exact diagonalization is the asymptotic bottleneck in the extremal case. Since the projection is a $\abs{B}\times\abs{B}$ matrix, constructing it densely from a sparse representation takes time $\mathcal{O}(\abs{B}^2) = \mathcal{O}(d^{2r})$. Then, exact diagonalization of the dense projection takes time $\mathcal{O}(\abs{B}^3) = \mathcal{O}(d^{3r})$. In the extremal case, subsequently identifying the solution is easy: the algorithm simply returns the support of one of the obtained extremal eigenvectors, so the total runtime is $\mathcal{O}(d^{3r})=\mathcal{O}(d^{3s-3})$.

    However, in the non-extremal setting, identifying the solution turns out to be more expensive than the diagonalization step, whose runtime is $\mathcal{O}(d^{3r}) = \mathcal{O}(d^{6s-6})$. We now analyze the cost of computing the energy variance for every eigenvector of the projection. Each eigenvector $\ket{\phi}$ is $\abs{B}$-sparse and the full Hamiltonian $H$ has at most $d$ nonzero entries per column, hence $H\ket{\phi}$ is $(d\abs{B})$-sparse. If the entries of $H\ket{\phi}$ are naively stored in an unordered list, then the sparse matrix-vector multiplication requires $\mathcal{O}(d^2\abs{B}^2)$ time in the worst case. Afterwards, computing $\bra{\phi}H^2\ket{\phi} = \norm{H\ket{\phi}}^2$ and $\bra{\phi}H\ket{\phi}$ take time $d\abs{B}$ and $d\abs{B}^2$, respectively. Hence, computing the variance for all $\mathcal{O}(\abs{B})$ eigenvectors of the projection takes time $\mathcal{O}(d^2\abs{B}^2\cdot \abs{B})=\mathcal{O}(d^2\abs{B}^3) = \mathcal{O}(d^{6s-4})$. Note that the use of more sophisticated data structures can decrease the expected runtime of this step significantly. 

    Nonetheless, the runtime of Subroutine~\ref{subroutine:sparse} dominates the variance calculation step, even when the latter is implemented naively. There are $\binom{N}{s}$ choices of the subset $S$, and computing the rank of each $V_{S^c}$ takes $\mathcal{O}(NM\min\{N,M\})$ time via Gaussian elimination (the single computation of $\rank V$ in the beginning also takes the same time). Hence, identifying whether there exists a subset $S$ which supports the desired $s$-sparse vector takes time $\mathcal{O}\left(\binom{N}{s}NM\min\{N,M\}\right)=\mathcal{O}\left(N^sNM\min\{N,M\}\right)=\mathcal{O}\left(N^{s+1}M\min\{N,M\}\right)$. Once such an $S$ has been identified, it is possible to find a $w$ which is in $\ker V_{S^c}\setminus \ker V$ by performing Gaussian elimination on $V_{S^c}$ in time $\mathcal{O}(NM\min\{N,M\})$. This computation occurs only once, so the runtime for the full subroutine is $\mathcal{O}\left(N^{s+1}M\min\{N,M\}\right)+\mathcal{O}(NM\min\{N,M\}) = \mathcal{O}\left(N^{s+1}M\min\{N,M\}\right)$.

    Lastly, when this subroutine is deployed in Algorithm~\ref{algorithm:eigenwalk}, $N = \abs{B}$ and the subroutine is called once per distinct eigenvalue $\mu$ of $\Pi_B H \Pi_B$, with $M = \abs{R_\mu}$ eigenvectors of the full Hamiltonian belonging to that eigenspace. Since every $R_\mu\leq \abs{B}$ and ${\sum_\mu \abs{R_\mu} \leq \abs{B}}$, the total cost across all calls is $\sum_\mu \mathcal{O}(\abs{B}^{s+1}\abs{R_\mu}^2)= \mathcal{O}(\abs{B}^{s+1}\sum_\mu\abs{R_\mu}^2)=\mathcal{O}(\abs{B}^{s+1}(\sum_\mu\abs{R_\mu})^2)= \mathcal{O}(\abs{B}^{s+3}) = \mathcal{O}(d^{r(s+3)}) = \mathcal{O}(d^{(2s-2)(s+3)})$. The cost of this step subsumes all previous steps for $s\geq 2$, so the total runtime for Algorithm~\ref{algorithm:eigenwalk} in the non-extremal case is $\mathcal{O}(d^{(2s-2)(s+3)})$.

\end{proof}

\subsection{Rayleigh--Ritz Theorem}\label{appendix:rayleigh}
Below we provide a self-contained statement of the Rayleigh--Ritz Theorem \cite{horn2012matrix} used in the proof of Theorems~\ref{thm:degeneracy}, \ref{thm:alg}, and \ref{thm:spca}.
\begin{theorem*}[Rayleigh--Ritz]
    Given a Hermitian operator $H$ on an $N$-dimensional Hilbert space, let $\lambda_1$ be the lowest (or highest) eigenvalue. If a normalized vector $\ket{\psi}$ satisfies 
    \begin{align}
        \bra{\psi}H\ket{\psi}=\lambda_1,
    \end{align}
    then $\ket{\psi}$ is an eigenvector of $H$ with eigenvalue $\lambda_1$.
\end{theorem*}
\begin{proof}
    Assume without loss of generality that the eigenvalues $\lambda_j$ of $H$ are ordered such that $\lambda_1\leq\lambda_2\leq\ldots\leq \lambda_N$. Since $H$ is Hermitian, it has an orthonormal basis of eigenvectors $\ket{\psi_j}$ corresponding to the eigenvalues $\lambda_j$ for $j=1,\ldots, N$. Write $\ket{\psi} = \sum_jc_j\ket{\psi_j}$ for some coefficients $c_j$ such that $\sum_j\abs{c_j}^2 = 1$. Then $\bra{\psi}H\ket{\psi} = \sum_j\abs{c_j}^2\lambda_j$. Since every $\lambda_j\geq \lambda_1$,
    \begin{align}
        \bra{\psi}H\ket{\psi} = \sum_j\abs{c_j}^2\lambda_j \geq \sum_j\abs{c_j}^2 \lambda_1 =\lambda_1
    \end{align}
    with equality if and only if $c_j = 0$ for all $j$ such that $\lambda_j > \lambda_1$. Hence, $\ket{\psi}$ must only have support on the eigenspace of $H$ corresponding to eigenvalue $\lambda_1$, from which it follows that $H\ket{\psi}= \lambda_1\ket{\psi}$.
\end{proof}

\section{Additional eigenwalk problem examples}\label{appendix:examples}
In this appendix, we provide additional graph examples that complement results in the main text.
Section~\ref{appendix:examples12} constructs explicit examples showing that the sufficient conditions for degeneracy in Theorems~\ref{thm:degeneracy} and~\ref{thm:eigstruct} are tight, and demonstrates that these theorems admit no converse.
Section~\ref{appendix:examples3} constructs path-graph examples that saturate both the extremal and non-extremal bounds in Theorem~\ref{thm:localized}, confirming that these bounds cannot be improved in general.

\subsection{Examples for Theorems \ref{thm:degeneracy} and \ref{thm:eigstruct}} \label{appendix:examples12}
\begin{figure}[h]
    \centering
    \includegraphics[width = 0.55\linewidth]{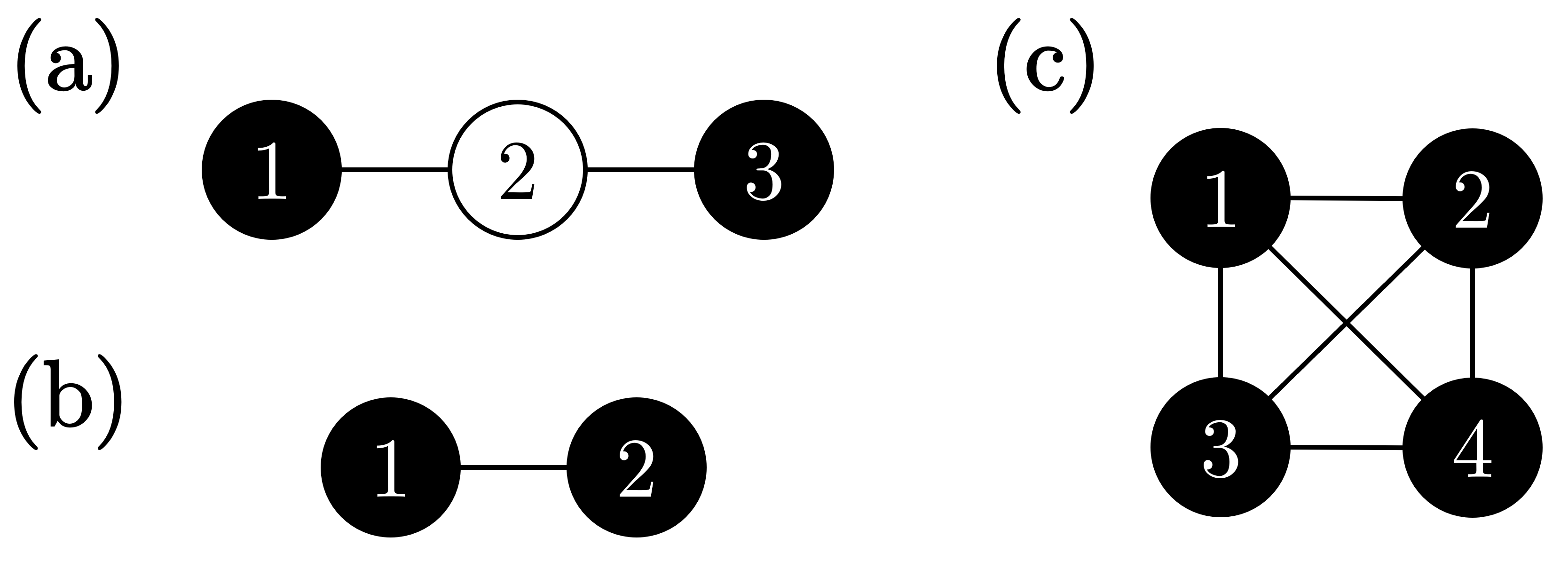}
    \caption{More graph examples, each with a subset $S$ of shaded vertices. For the graph in \textbf{(a)}, there exists a graph-consistent Hermitian matrix $H$ admitting a nondegenerate non-extremal eigenvector whose support is $S$. For the graph in \textbf{(b)}, there exists a graph-consistent Hermitian matrix $H$ admitting a nondegenerate extremal eigenvector whose support is $S$. The graph in \textbf{(c)} admits a graph-consistent Hermitian matrix $H$ with a degenerate eigenvector whose support cannot be partitioned into the supports of eigenvectors with the same eigenvalue.} 
    \label{fig:more-examples}
\end{figure}

We first provide some explicit examples to confirm that the sufficient conditions of Theorem \ref{thm:degeneracy} are tight. Namely, there exist nondegenerate non-extremal (resp. extremal) eigenvectors whose support can be partitioned into $S,S'$ separated by a distance of 2 (resp. 1) in the underlying graph. Consider the three-vertex path graph depicted in Figure~\ref{fig:more-examples}a and its adjacency matrix
\begin{align} \label{eq:3-v-path}
    H = \begin{pmatrix}
        0 & 1 & 0\\
        1 & 0 & 1\\
        0 & 1 & 0
    \end{pmatrix},
\end{align}
which is consistent with the graph. $H$ has spectrum $\{-\sqrt{2},0,+\sqrt{2}\}$, and the nondegenerate (non-extremal) eigenvector associated with eigenvalue $0$ is $\ket{1} - \ket{3}$. The vertices $\{1,3\}$ in the support are separated by a distance of 2, which implies that the non-extremal bound in Theorem~\ref{thm:degeneracy} is tight.

Similarly, consider the connected two-vertex graph depicted in Figure~\ref{fig:more-examples}b and its adjacency matrix
\begin{align}
    H = \begin{pmatrix}
        0 & 1\\
        1 & 0
    \end{pmatrix},
\end{align}
which is also graph-consistent. $H$ has spectrum $\{-1,+1\}$, and the nondegenerate (extremal) eigenvector associated with eigenvalue $-1$ is $\ket{1}-\ket{2}$. The vertices $\{1,2\}$ in the support are separated by a distance of $1$, which implies that the extremal bound in Theorem~\ref{thm:degeneracy} is tight as well.

Next, we demonstrate that Theorem~\ref{thm:degeneracy} has no converse, that is, given two degenerate eigenvectors whose supports are disjoint sets, these sets need not be separated by a distance $>1$ in the graph. We now present an example where these sets are separated by a distance of exactly 1. Consider the four-vertex complete graph in Figure~\ref{fig:more-examples}c, and define $H(c)$ to be the graph-consistent family of matrices
\begin{align}
    H(c) = \begin{pmatrix}
        0 & 1 & c & c\\
        1 & 0 & c & c\\
        c & c & 0 & 1\\
        c & c & 1 & 0 
    \end{pmatrix},\qquad
    c \neq 0.
\end{align}
The spectrum of $H(c)$ is $\{-1, -1, 1+2c, 1-2c\}$, and $\ket{1}-\ket{2}$ and $\ket{3}-\ket{4}$ are degenerate eigenvectors with eigenvalue $-1$. Moreover, the vertex sets $\{1,2\}$ and $\{3,4\}$ are separated by a distance of 1, as desired. Note that these $-1$ eigenvectors can be made either extremal or non-extremal depending on the choice of $c$: for instance, if $c = 1/2$, they are ground states, and if $c = 2$, they are non-extremal.

We also provide an explicit example of the virtual vertex phenomenon discussed after the statement of Theorem~\ref{thm:eigstruct}. Recall the graph of Figure~\ref{fig:more-examples}a and its adjacency matrix $H$ in Eq.~(\ref{eq:3-v-path}), which has nondegenerate non-extremal eigenvector $\ket{1}-\ket{3}$. The support vertices 1 and 3 are connected only via a path that crosses a vertex outside the support (vertex 2). At the same time, the eigenvector cannot be decomposed into sparser eigenvectors with the same eigenvalue, since it is nondegenerate. 

\subsection{Examples for Theorem \ref{thm:localized}}\label{appendix:examples3}
We now construct examples which saturate the bound in Theorem~\ref{thm:localized}. Let $G$ be the path graph shown in both parts of Figure~\ref{fig:paths}, and suppose $G$ contains some odd number $m$ of vertices. Now define $H$ to be the adjacency matrix of $G$. It is known \cite{behn2011} (and straightforward to verify) that the eigenvalues and corresponding eigenvectors of $H$ are 
\begin{align}
    \lambda_k = 2\cos(\frac{k\pi}{m+1}),\quad  \ket{\psi_k} = \sum_{j=1}^m \sin(\frac{jk\pi}{m+1})\ket{j}
\end{align}
for $k=1,\ldots,m$. Evidently, all of the eigenvectors are nondegenerate. 

Consider the top eigenvector $\ket{\psi_1}$. Since 
$0<j\pi/(m+1) < \pi$ for $j=1,\ldots, m$, $\ket{\psi_1}$ is supported on all the vertices in $G$. Since $G$ is a path, the diameter of $\supp(\ket{\psi_1})$ is $s-1$, where $s = \abs{\supp(\ket{\psi_1})}$. Furthermore, since $\ket{\psi_1}$ is nondegenerate, it is the only top
eigenvector up to scalar multiplication. Therefore, for any
$v\in\supp(\ket{\psi_1})$, every top eigenvector whose support contains $v$
has support diameter $s-1$. This example thus saturates the extremal bound in Theorem~\ref{thm:localized}, showing that the bound is tight.

Next consider the non-extremal eigenvector $\ket{\psi_\mu}$, where $\mu = (m+1)/2$ (recall that $m$ was chosen to be odd, so $\mu$ is an integer). It is readily checked that the coefficients of $\ket{\psi_\mu}$ are zero on even basis states and nonzero on odd states. Hence, $s=\abs{\supp(\ket{\psi_\mu})}$ is equal to the number of odd states, which is $\mu$. The endpoints of $G$ are contained within $\supp(\ket{\psi_\mu})$ and are separated by a distance of $m-1 = 2\mu-2 = 2s-2$. Thus, $\supp(\ket{\psi_\mu})$ has diameter $2s-2$. Furthermore, since $\ket{\psi_\mu}$ is nondegenerate, it is the only eigenvector with eigenvalue $\lambda_\mu$, up to scalar multiplication. Therefore, for any
$v\in\supp(\ket{\psi_\mu})$, every eigenvector with eigenvalue $\lambda_\mu$ whose support contains $v$
has support diameter $2s-2$. This example thus saturates the non-extremal bound in Theorem~\ref{thm:localized}, showing that this bound is also tight.

\section{Application to sparse principal component analysis}\label{appendix:spca}
Given a covariance matrix $\Sigma$, sparse principal component analysis (SPCA) seeks an $s$-sparse vector $\ket{\psi}$ that maximizes the variance $\bra{\psi}\Sigma \ket{\psi}$, where $s$ is a given positive integer \cite{spca}. Analogously to solving for exact eigenvectors, the task of computing such a $\ket{\psi}$ can be reduced to finding its support, which resembles the eigenwalk problem. However, since $\ket{\psi}$ need not be an exact eigenvector of $\Sigma$, it is not obvious that the support-localization result of Theorem~\ref{thm:localized} must apply; hence, it is not clear \textit{a priori} that that the support of $\ket{\psi}$ can be found efficiently by an eigenwalk. 

The following theorem generalizes Theorem~\ref{thm:localized} to the setting of SPCA, thereby proving that the supports of sparse principal components are indeed tightly localized in the graph determined by the nonzero entries of the covariance matrix. 
\setcounter{theorem}{5}
\begin{theorem}[Support localization for sparse principal components]\label{thm:spca}
    Let $G$ be a simple graph and $\Sigma$ be a $G$-consistent Hermitian matrix. Given a positive integer $s$, suppose the normalized state $\ket{\psi}$ belongs to the solution set $\mathcal{P}$ of the following SPCA problem:
    \begin{align}\label{eq:spca}
        \mathcal{P} = \argmax_{\norm{\phi}=1,\  \abs{\supp(\ket{\phi})} \leq s} \bra{\phi}\Sigma\ket{\phi}.
    \end{align}
    Then, given any vertex $v\in\supp(\ket{\psi})$, there exists a state $\ket{\psi_v}$ with $v\in\supp(\ket{\psi_v})$ such that $\ket{\psi_v} \in\mathcal{P}$ and 
    \begin{align}
        \operatorname{diam}_G\bigl(\supp(\ket{\psi_v})\bigr) \leq s-1.
    \end{align}
    
\end{theorem}
\begin{proof}
    Let $R = \supp(\ket{\psi})$. Since $\Pi_R \ket{\psi} = \ket{\psi}$, we have $\bra{\psi}\Pi_R\Sigma\Pi_R\ket{\psi} = \bra{\psi}\Sigma\ket{\psi}$. It must be true that $\ket{\psi}$ maximizes $\bra{\phi}\Pi_R\Sigma\Pi_R\ket{\phi}$ over all normalized states $\ket{\phi}$ in the span of $R$, otherwise $\ket{\psi}$ would not belong to $\mathcal{P}$. It follows from the Rayleigh--Ritz Theorem (see Appendix \ref{appendix:rayleigh}) that $\ket{\psi}$ is a highest eigenvector of $\Pi_R\Sigma\Pi_R$ restricted to the span of $R$.
    
    Now regard $\Pi_R\Sigma\Pi_R$ restricted to the span of $R$ as a Hermitian matrix consistent with the induced subgraph $G[R]$, that is, the graph with vertex set $R$ and whose edge set contains all edges of $G$ with both endpoints in $R$. By Theorem~\ref{thm:localized}, for any $v\in\supp(\ket{\psi})$, there exists a normalized highest eigenvector $\ket{\psi_v}$ of the restricted $\Pi_R\Sigma\Pi_R$ such that $v\in\supp(\ket{\psi_v})$ and
    \begin{align}
        \operatorname{diam}_{G[R]}\bigl(\supp(\ket{\psi_v})\bigr)
        \leq \abs{\supp(\ket{\psi})}-1\leq s-1.
    \end{align}
    Since distances in $G$ are no larger than distances in the induced subgraph $G[R]$, it follows that
    \begin{align}
        \operatorname{diam}_{G}\bigl(\supp(\ket{\psi_v})\bigr)\leq s-1.
    \end{align}
    Moreover, since $\supp(\ket{\psi_v})\subseteq R$, we have
    \begin{align}
        \bra{\psi_v}\Sigma\ket{\psi_v}
        =
        \bra{\psi_v}\Pi_R\Sigma\Pi_R\ket{\psi_v}
        =
        \bra{\psi}\Pi_R\Sigma\Pi_R\ket{\psi}
        =
        \bra{\psi}\Sigma\ket{\psi}.
    \end{align}
    Hence, $\ket{\psi_v}\in\mathcal{P}$, as desired.
\end{proof}

Theorem~\ref{thm:spca} implies that SPCA might be accelerated by an eigenwalk when a guiding vertex, promised to belong to the support of an $s$-sparse principal component, is known. In this case, a graph search can be used to collect all vertices within radius $s-1$ of the guiding vertex. The covariance matrix can then be projected to the span of the collected vertices, and the same $s$-sparse PCA problem can be solved on this projection. By Theorem~\ref{thm:spca}, there exists an optimizer of the original SPCA problem whose support is contained in the ball of radius $s-1$ around the guiding vertex. Therefore, the SPCA problem restricted to this ball has the same optimum value as the original problem. Since every feasible vector for the restricted problem is also feasible for the original problem, any optimizer of the restricted problem is also an optimizer of the original problem. Ideally, the projected covariance matrix is much smaller than the full covariance matrix, making the subsequent SPCA instance more efficient. However, this dimensional reduction is only useful when the collected ball is a proper subset of the full vertex set. Thus the covariance matrix and its underlying graph must have nontrivial sparsity structure; for example, if the covariance matrix is fully dense, then the graph is a complete graph and every radius-one ball already contains all vertices.

\newpage
\bibliographystyle{quantum}
\bibliography{References}
\end{document}